%% file: arxiv_final.tex
\documentclass[journal]{IEEEtran}
\usepackage{amsthm,amsmath,amssymb,color,cleveref,cite}
\usepackage{mathrsfs}
\usepackage[ruled,vlined]{algorithm2e}
\usepackage{caption}
\usepackage{thmtools}
\usepackage{thm-restate}

\input{defns}
\usepackage{amsfonts}
\usepackage{graphics}
\usepackage{dsfont}

\newtheorem{thm}{Theorem}

\newtheorem{prop}{Proposition}
\newtheorem{lem}{Lemma}

\def\erdos{Erd\H os}
\def\poi{\mathrm{Poi}}
\def\renyi{R\'enyi }
\def\dtv{d_\mathrm{TV}}
\theoremstyle{definition}

\newcommand{\indi}{\mathbbm{1}}

\newtheorem{rem}{Remark}

\newcommand{\BB}{{\bf B}}
\newcommand{\Qv}{{\bf Q}}
\newcommand{\pv}{{\bf p}}
\newcommand{\Thetav}{{\bf \Theta}}

\newcommand{\Lap}{\mathrm{L}}
\newcommand{\bern}{\mathrm{Bern}}
\def\Gnm{{\mathcal{G}_{n,m}}}

\def\Qm{Q_\mathrm{max}}
\def\E{{\mathsf E}}
\def\P{{\mathsf P}}
\def\Gnm{{\mathcal{G}_{n,m}}}
\def\lowerBC{{\underline{\Sigma}}}
\def\upperBC{{\overline{\Sigma}}}
\def\PQP{\mathbf{p}^T\mathbf{Q}\mathbf{p}}
\def\PsQP{\mathbf{p}^T\mathbf{Q^*}\mathbf{p}}
\def\PhQP{\mathbf{p}^Th(f(n)\mathbf{Q})\mathbf{p}}

\DeclarePairedDelimiter\ceil{\lceil}{\rceil}
\DeclarePairedDelimiter\floor{\lfloor}{\rfloor}

\title{Universal Graph Compression:\\
Stochastic Block Models}

\author{Alankrita~Bhatt$^*$,
        Ziao~Wang$^*$,
	Chi~Wang,
        and~Lele~Wang
\thanks{This work was supported in part by the NSF Center for Science of Information postdoctoral fellowship, in part by the NSERC Discovery Grant RGPIN-2019-05448, and in part by the NSERC Discovery Launch Supplement DGECR-2019-00447. This work was presented in part at the 2021 IEEE International Symposium on Information Theory.}
\thanks{$*$ Alankrita Bhatt and Ziao Wang contributed equally to this work.}
\thanks{Alankrita Bhatt is with the California Institute of Technology, Pasadena, CA, 91125, USA (email: abhatt@caltech.edu).}
\thanks{Ziao Wang is with the Department of Electrical and Computer Engineering, University of British Columbia, Vancouver, BC V6T1Z4, Canada (email: ziaow@ece.ubc.ca).}
\thanks{Chi Wang is with Microsoft Research, Redmond, WA 98052, USA (wang.chi@microsoft.com).}
\thanks{Lele Wang is with the Department of Electrical and Computer Engineering, University of British Columbia, Vancouver, BC V6T1Z4, Canada (email: lelewang@ece.ubc.ca).}}

\begin{document}
\captionsetup[figure]{labelfont={rm},labelformat={default},labelsep=period,name={Fig.}}
\captionsetup[table]{labelformat={default},font={sc},name={TABLE},labelsep={none}}

\maketitle
\begin{abstract}
Motivated by the prevalent data science applications of processing large-scale graph data such as social networks and biological networks, this paper investigates lossless compression of data in the form of a labeled graph. Particularly, we consider a widely used random graph model, stochastic block model (SBM), which captures the clustering effects in social networks. An information-theoretic universal compression framework is applied, in which one aims to design a \emph{single} compressor that achieves the asymptotically optimal compression rate, for every SBM distribution, without knowing the parameters of the SBM. Such a graph compressor is proposed in this paper, which universally achieves the optimal compression rate with polynomial time complexity for a wide class of SBMs.
Existing universal compression techniques are developed mostly for \emph{stationary ergodic} one-dimensional sequences. However, the adjacency matrix of SBM has complex two-dimensional correlations. The challenge is alleviated through a carefully designed transform that converts two-dimensional correlated data into \emph{almost} i.i.d. submatrices. The sequence of submatrices is then compressed by a Krichevsky--Trofimov compressor, whose length analysis is generalized to identically distributed but arbitrarily correlated sequences. In four benchmark graph datasets, the compressed files from competing algorithms take 2.4 to 27 times the space needed by the proposed scheme.
\end{abstract}

\section{Introduction}
\IEEEPARstart{I}{n} many data science applications, data appears in the form of large-scale graphs. For example, in social networks, vertices represent users and an edge between vertices represents friendship; in the World Wide Web, vertices are websites and edges indicate the hyperlinks from one site to the other; in biological systems, vertices can be proteins and edges illustrate protein-to-protein interaction. Such graphs may contain billions of vertices. In addition, edges tend to be correlated with each other since, for example, two people sharing many common friends are likely to be friends as well. How to efficiently compress such large-scale structural information to reduce the I/O and communication costs in storing and transmitting such data is a persisting  challenge in the era of big data.

The problem of graph compression has been studied in different communities, following various different methodologies. Several methods exploited combinatorial properties such as cliques and cuts in the graph~\cite{Rossi--Zhou2018,Lim--Kang--Faloutsos2014}. Many works targeted at domain-specific graphs such as web graphs~\cite{Boldi--Vigna2004}, biological networks~\cite{Conway--Thomas--Bromage--Andrew2011,Hayashida--Akutsu2010}, and social networks~\cite{Chierichetti--Kumar--Lattanzi--Mitzenmacher--Panconesi--Raghavan2009}. Various representations of graphs were proposed, such as the text-based method, where the neighbour list of each vertex is treated as a ``word''~\cite{Navarro2007,Sadakane2003}, and the $k^2$-tree method, where the adjacency matrix is recursively partitioned into $k^2$ equal-size submatrices~\cite{Brisaboa--Ladra--Navarro2009}.  \emph{Succinct} graph representations that enable certain types of fast computation, such as adjacency query or vertex degree query, were also widely studied~\cite{Farzan--Munro2013}. We refer the readers to~\cite{Besta--Hoefler2018} for a survey on lossless graph compression and space-efficient graph representations.

In this paper, we take an information theoretic approach to study lossless compression of graphs with vertex labels.
We assume the graph is generated by some random graph model and investigate lossless compression schemes that achieve the theoretical limit, i.e., the entropy of the graph, asymptotically as the number of vertices goes to infinity. When the underlying distribution/statistics of the random graph model is known, optimal lossless compression can be achieved by methods like Huffman coding. However, in most real-world applications, the exact distribution is usually hard to obtain and the data we are given is a single realization of this distribution. This motivates us to consider the framework of \emph{universal compression}, in which we assume the underlying distribution belongs to a known family of distributions and require that the encoder and the decoder should not be a function of the underlying distribution. The goal of universal compression is to design a single compression scheme that universally achieves the optimal theoretical limit, for every distribution in the family, without knowing which distribution generates the data. For this paper, we focus on the family of \emph{stochastic block models}, which are widely used random graph models that capture the clustering effect in social networks. Our goal is to develop a universal graph compression scheme for a family of stochastic block models with as wide a range of parameters as possible.

How to design a computationally efficient universal compression scheme is a fundamental question in information theory. In the past several decades, a large number of universal compressors were proposed for one-dimensional sequences with fixed alphabet size, whose entropy is linear in the number of variables. Prominent results include the Laplace  and Krichevsky--Trofimov (KT) compressors for i.i.d.\! processes~\cite{Laplace1995, Krichevsky--Trofimov1981}, Lempel--Ziv compressor~\cite{Lempel--Ziv77,Lempel--Ziv78} and Burrows--Wheeler transform\cite{effros2002universal} for stationary ergodic processes, and context tree weighting~\cite{willems1995context} for finite memory processes. Many of these have been adopted in standard data compression applications such as \texttt{compress}, gzip, GIF, TIFF, and bzip2. Despite these exciting developments, existing universal compression techniques fall short of establishing optimality results for graph data due to the following challenges. Firstly, graph data generated from a stochastic block model has non-stationary two-dimensional correlations, so existing techniques do not immediately apply here. Secondly, in many practical applications, where the graph is sparse, the entropy of the graph may be sublinear in the number of entries in the adjacency matrix.

For the first challenge, a natural question arising is: can we convert the two-dimensional adjacency matrix of the graph into a one-dimensional sequence in some order and apply a universal compressor for the sequence? For some simple graph models such as \erdos--\renyi graph, where each edge is generated i.i.d.\! with probability $p$, this would indeed work. For more complex graph models including stochastic block models, it is unclear whether there is an ordering of the entries that results in a stationary process. We will show in Section~\ref{sec:C1discussion}~\footnote{See also Section~\ref{sec:C1discussion} for a discussion on why several natural encoding schemes such as first performing community detection and then compressing, or run-length based coding schemes also do not work in the current setting.} several orders including row-by-row, column-by-column, and diagonal-by-diagonal fail to produce a stationary process.

We alleviate these challenges by designing a decomposition of the adjacency matrix into submatrices. We then show in Proposition~\ref{prop:blkIndep} that with a carefully chosen parameter, the submatrix decomposition converts two-dimensional correlated entries into a sequence of \emph{almost} i.i.d.\! submatrices with slowly growing alphabet size. To address the second challenge, we introduce a new definition of universality that accommodates data with an unknown leading order in its entropy expression. The new universality normalizes the compression length by the entropy of data, requiring it to go to one, as compared to normalizing by the number of variables in the standard definition. Normalizing by the number of variables, implicitly assumes that the entropy is linear in $n$.

{\bf Notation.} 
For an integer $n$, let $[n] = \{1,2,\ldots,n\}$. Let $\log(\cdot) = \log_2(\cdot)$. We follow the standard order notation: $f(n) = O(g(n))$ if $\lim_{n \to \infty} \frac{|f(n)|}{g(n)} < \infty$; $f(n) = \Omega(g(n))$ if $\lim_{n \to \infty}\frac{f(n)}{g(n)} >0$; $f(n) = \Theta(g(n))$ if $f(n) = O(g(n))$ and $f(n) = \Omega(g(n))$; $f(n) = o(g(n))$ if $\lim_{n \to \infty}\frac{f(n)}{g(n)} = 0$; $f(n) = \omega(g(n))$ if $\lim_{n \to \infty}\frac{|f(n)|}{|g(n)|} = \infty$; and $f(n) \sim g(n)$ if $\lim_{n \to \infty}\frac{f(n)}{g(n)} = 1$. For $0 \le p \le 1$, let $h(p) = -p\log(p) - (1-p) \log(1-p)$ denote the binary entropy function. For a matrix $W$ with entries in $[0,1]$, let $h(W)$ be a matrix of the same dimension whose $(i,j)$ entry is $h(W_{ij})$.

\subsection{Problem Setup}\label{sec:setup}
{\bf Universal graph compressor.} For simplicity, we focus on simple (undirected, unweighted, no self-loop) graphs with labeled vertices in this paper. But our compression scheme and the corresponding analysis can be extended to more general graphs. Let $\Ac_n$ be the set of all labeled simple graphs on $n$ vertices. Let $\{0,1\}^i$ be the set of binary sequences of length $i$, and set $\{0,1\}^* = \cup_{i=0}^\infty\{0,1\}^i$. A lossless graph compressor $C\suchthat \Ac_n \to \{0,1\}^*$ is a one-to-one function that maps a graph to a binary sequence. Let $\ell(C(A_n))$ denote the length of the output sequence. When $A_n$ is generated from a distribution, it is known that the entropy $H(A_n)$  is a fundamental lower bound on the expected length of any lossless compressor~\cite[Theorem 10.5]{Polyanskiy--Wu14}
\begin{align}\label{eq:losslessLB}
    H(A_n)-\log(e(H(A_n)+1)) \le \E[\ell(C(A_n))],
\end{align} 
and therefore 
\[
 \liminf_{n \to \infty}\frac{\E[\ell(C(A_n))]}{H(A_n)} \ge 1.
\]
Thus, a graph compressor is said to be \emph{universal} for the family of distributions $\mathscr{P}$ if for all distribution $\P \in \mathscr{P}$ and $A_n \sim \P$, we have
\begin{equation}
 \limsup_{n \to \infty}\frac{\E[\ell(C(A_n))]}{H(A_n)} = 1.
\end{equation}

{\bf Stochastic block model.} A stochastic block model $\mathrm{SBM}(n,L,\pv,\Wv)$ defines a probability distribution over $\Ac_n$. Here $n$ is the number of vertices, $L$ is the number of communities. Each vertex $i \in [n]$ is associated with a community assignment $X_i \in [L]$. The length-$L$ column vector $\pv=(p_1,p_2,\ldots,p_L)^T$ is a probability distribution over $[L]$, where $p_i$ indicates the probability that any vertex is assigned community $i$. $\Wv$ is an $L\times L$ symmetric matrix, where $W_{ij}$ represents the probability of having an edge between a vertex with community assignment $i$ and a vertex with community assignment $j$. We say $A_n \sim \mathrm{SBM}(n,L,\pv,\Wv)$ if the community assignments $X_1,X_2,\ldots,X_n$ are generated i.i.d.\! according to $\pv$ and for every pair $1 \le  i<j \le n$, an edge is generated between vertex $i$ and vertex $j$ with probability $W_{X_i,X_j}$. In other words, in the adjacency matrix $A_n$ of the graph, $A_{ij} \sim \Bern(W_{X_i,X_j})$ for $i <j$; the diagonal entries $A_{ii} = 0$ for all $i \in [n]$; and $A_{ij} = A_{ji}$ for $i > j$. We write $\Wv = f(n)\Qv$, where $\Qv$ is an $L\times L$ symmetric matrix with $\max_{i,j}Q_{i,j}=\Theta(1)$. We assume all entries in $\pv$ are non-zero constants, and $L = \Theta(1)$. Furthermore, to make sure that the model is non-trivial, we assume that all entries in $W$ are at most $1$, and at least one entry is strictly less than $1$. Also, because we can always scale the function $f(n)$, we assume without generality that $\max_{i,j}Q_{i,j}<1$. We will consider two families of stochastic block models: For $0 < \e < 1$,
\begin{align}
 \mathscr{P}_1(\eps)\suchthat &\mathrm{SBM}(n,L,\pv,\Wv)\nonumber\\ &\text{ with } f(n) = O(1), f(n)= \Omega\left(\frac{1}{n^{2-\eps}}\right),\label{eq:P1family}\\
 \mathscr{P}_2(\eps)\suchthat &\mathrm{SBM}(n,L,\pv,\Wv)\nonumber \\ &\text{ with } f(n) = o(1), f(n)= \Omega\left(\frac{1}{n^{2-\eps}}\right).\label{eq:P2family}
\end{align}
Note that the edge probability $\frac{1}{n^2}$ is the threshold for a random graph to contain an edge with high probability~\cite{Frieze--Karonski2015}. Thus, the family $\mathscr{P}_1(\e)$ covers most non-trivial SBM graphs. Clearly, $\mathscr{P}_2(\eps)$ is a strict subset of $\mathscr{P}_1(\eps)$, as it does not contain the constant regime $f(n) = 1$.

{\bf Minimax redundancy.} Besides universality, which only concerns the first order terms in the expected length of the compressor and the entropy of the graph, the redundancy of a universal compressor is another important metric we wish to optimize in the design of universal graph compressors. Define the redundancy of a lossless compressor $C$ for $A_n \sim \P$ as
\begin{equation}
R(C,A_n) = \E[\ell(C(A_n))] - H(A_n).
\end{equation}
Then the minimax redundancy over a distribution family $\mathscr{P}$ is defined as
\begin{equation}
    R^*_n(\mathscr{P}) = \inf_{C} \sup_{\P \in \mathscr{P}} R(C,A_n).
\end{equation}
The minimax redundancy represents the lowest achievable redundancy of any lossless compressor for the family $\mathscr{P}$. 
Define the family of SBM in regime $f(n)$ with maximum entry $Q_\mathrm{max}$ in $\Qv$ as
\begin{align}
    \mathscr{P}_3(f(n),Q_\mathrm{max}): &\mathrm{SBM}(n,L,\pv,f(n)\Qv)\nonumber\\  &\text{ with } Q_{i,j} \le Q_\mathrm{max} \text{ for all } i,j \in [n].
\end{align}
We comment that if $f(n)$ satisfies $f(n)=o(1)$ and $f(n)=\Omega(\frac{1}{n^{2-\epsilon}})$ for some $0<\epsilon<1$, then $\mathscr{P}_3(f(n),Q_\mathrm{max})$ is a subset of $\mathscr{P}_2(\epsilon)$ with fixed regime $f(n)$ and entries in $\Qv$ upper bounded by $Q_\mathrm{max}$.  Hence, it is also a subset of $\mathscr{P}_1(\epsilon)$. In this work, we will study the minimax redundancy over the family $\mathscr{P}_3(f(n),Q_\mathrm{max})$ instead of $\mathscr{P}_1(\epsilon)$ or $\mathscr{P}_2(\epsilon)$. The reason is that both the expected length of the compressor and the entropy of the graph scale with $f(n)$ and $Q_\mathrm{max}$. By fixing $f(n)$ and $Q_\mathrm{max}$ instead of taking supremum over them, we can better understand the minimax redundancy of SBMs in each regime.

\subsection{Literature in the information-theoretic framework}
Under the information-theoretic framework, lossless graph compression has been studied in~\cite{Turan1984,Naor1990,Lempel--Ziv86,Choi--Szpankowski2012,Abbe16,Asadi--Abbe--Verdu2017,Kieffer--Yang--Szpankowski2009,Golundefinedbiewski--Magner--Szpankowski2018,Magner--Turowski--Szpankowski2018,Luczak--Magner--Szpankowski2019,Turowski--Magner--Szpankowski2018,Kontoyiannis--Lim--Papakonstantinopoulou--Szpankowski2021,Zhang--Yang--Kieffer2014,Ganardi--Hucke--Lohrey--Seelbach-Benkner2019,Delgosha--Anatharam2019tit,Delgosha--Anatharam2020,Delgosha--Anatharam2021, Peixoto2013,Peixoto2014,Peixoto2015,Peixoto2017} and lossy compression of directed graphs was studied in~\cite{Bustin--Shayevitz21}.  In~\cite{Choi--Szpankowski2012,Kontoyiannis--Lim--Papakonstantinopoulou--Szpankowski2021}, universal compression of \emph{unlabeled} graphs (isomorphism classes) is investigated. For lossless compression of labeled graphs, \cite{Abbe16,Asadi--Abbe--Verdu2017} studied the compression of graphs generated from SBMs. However, it does not consider universal compression schemes. In~\cite{Lempel--Ziv86}, a universal compression algorithm is proposed for two-dimensional arrays of data by first converting the two-dimensional data into a one-dimensional sequence using a Hilbert--Peano curve and then applying the Lempel--Ziv algorithm~\cite{Lempel--Ziv78} for sequences.  In~\cite{Kieffer--Yang--Szpankowski2009,Zhang--Yang--Kieffer2014,Magner--Turowski--Szpankowski2018,Golundefinedbiewski--Magner--Szpankowski2018,Ganardi--Hucke--Lohrey--Seelbach-Benkner2019}, universal compression of tree graphs were studied.
In~\cite{Peixoto2013,Peixoto2014,Peixoto2015,Peixoto2017}, a minimum description length based method was proposed to infer network structures in stochastic block models. Recently, universal compression of graphs with marked edges and vertices was studied by Delgosha and Anantharam~\cite{Delgosha--Anatharam2019tit,Delgosha--Anatharam2020}. They focus on the \emph{sparse} graph regime, where the number of edges is in the same order as the number of vertices $n$.
They employ the framework of local weak convergence, which provides a technique to view a sequence of graphs as a sequence of distributions on neighbourhood structures. Built on this framework, they propose an algorithm that compresses graphs by describing the local neighbourhood structures. Moreover, they introduce a universality/optimality criterion through a notion of entropy for graph sequences under the local weak convergence framework, known as the \emph{Bordenave--Caputo (BC) entropy}~\cite{Bordenave_2014}. 
  In~\cite{Delgosha--Anatharam2021,Delgosha--Anantharam2021b}, a \emph{sparse graphon framework} is proposed for graphs with edge numbers in $\Theta(n^\alpha)$ for $1 < \alpha \le 2$. A new universality criterion is introduced and the corresponding universal compressor is proposed. 
The two universality criteria in~\cite{Delgosha--Anatharam2019tit,Delgosha--Anatharam2020, Delgosha--Anatharam2021,Delgosha--Anantharam2021b} are both stronger than the one used in this paper. They require the asymptotic length of the compressor to match the constants in both first and second-order terms in Shannon entropy, whereas the universality criterion we use only requires matching the first-order term. 
As a consequence of the stronger criteria, the compressors in~\cite{Delgosha--Anatharam2019tit,Delgosha--Anatharam2020, Delgosha--Anatharam2021,Delgosha--Anantharam2021b} are universal over a random graph family with a smaller range of edge-numbers. Moreover, the algorithm in~\cite{Delgosha--Anatharam2021,Delgosha--Anantharam2021b} is a combination of two different compressors, and has no polynomial-time implementation. The first compressor is universal over a family of graphs with $\Theta(n)$ edges that possess the so-called local weak convergence property. The second compressor is universal over the family of random graphs generated from sparse graphon with $\Theta(n^\alpha)$ edges for $1<\alpha<2$. We note that the sparse graphon framework includes stochastic block models with $f(n)=\omega(\frac{1}{n})$ and $f(n)=o(1)$ as a special case. In comparison, the compressor we propose is a single polynomial-time algorithm that applies to graphs with edge numbers in $\Theta(n^\alpha)$ for every $0< \alpha \le 2$.
 In Section~\ref{sec:BCentropy}, we evaluate the proposed compressor under the criterion in~\cite{Delgosha--Anatharam2019tit} for the family of stochastic block models. The proposed compressor achieves the same universality performance in terms of BC entropy. 

\subsection{Contribution and organization}

The main contributions of the paper are summarized as follows. Firstly, we propose a polynomial-time graph compressor and establish its universality over the corresponding family of
stochastic block models. Secondly, we present upper and lower bounds for the minimax redundancy of stochastic block models. Finally, we analyze the proposed compressor under the local weak convergence framework considered in~\cite{Delgosha--Anatharam2019tit,Bordenave_2014}. We show that the proposed compressor achieves the same performance guarantee (BC entropy) as the compressor in~\cite{Delgosha--Anatharam2019tit}.

The rest of the paper is organized as follows. In Section~\ref{sec:setup}, we defined universality over a family of graph distributions, the stochastic block models and the minimax redundancy of a family of distributions. We present our main result in Section~\ref{sec:main}, which is a polynomial-time graph compressor that is universal for a family containing most of the non-trivial stochastic block models. We describe the proposed graph compressor in Section~\ref{sec:univGraphComp}. In Section \ref{sec:maintheorems}, we state the main theorems of the paper---two theorems on 
the university of the proposed compressor
as well as a theorem on upper and lower bounds for the minimax redundancy. In Section~\ref{sec:experiments}, we implement our compressor in four benchmark graph datasets and compare its empirical performance to four competing algorithms. We illustrate key steps in establishing universality in Section~\ref{sec:theorem} and elaborate on the proof of each step in Section~\ref{sec:proposition}. In Section~\ref{sec:proveminimax}, we prove upper and lower bounds for the minimax redundancy. In Section~\ref{sec:BCentropy}, we provide the second-order analysis of the expected length of our compressor and compare it to the one in~\cite{Delgosha--Anatharam2019tit}. In Section~\ref{sec:C1discussion}, we explain why some natural attempts to compress stochastic block models are not applicable to our problem or not universal. In Section~\ref{sec:conclusion}, we conclude our results and introduce some open problems and potential future works.

\section{Main Results}
\label{sec:main}
We present our compression scheme in Section~\ref{sec:univGraphComp} and state its performance guarantee in Theorems~\ref{thm:strongUniv}~and~\ref{thm:weakUniv}. The key idea in our design is a decomposition of the adjacency matrix into submatrices, which, with a carefully chosen parameter, converts the two-dimensional correlated entries in the adjacency matrix into a sequence of \emph{almost} i.i.d.\! submatrices. We then compress the submatrices using a Krichevsky--Trofimov or Laplace compressor and generalize their length analyses from i.i.d.\! processes to \emph{arbitrarily correlated} but identically distributed processes.

\subsection{Algorithm: Universal Graph Compressor}\label{sec:univGraphComp}
In this section, we describe our universal graph compression scheme. 
For each integer  $k\le n$, the graph compressor $C_k\suchthat \Ac_n \to \{0,1\}^*$ is defined as follows.
\begin{itemize}
 \item {\bf Submatrix decomposition.}
Let $n' = \lfloor n/k \rfloor$ and $\tilde{n} = \lfloor n/k \rfloor k$. For $1 \le i,j \le n'$, let $\BB_{ij}$ be the submatrix of $A_n$ formed by the rows $(i-1)k + 1, (i-1)k + 2, \ldots, ik$ and the columns $(j-1)k + 1, (j-1)k + 2, \ldots, jk$. For example, we have
\begin{equation}\label{eq:B12}
\BB_{12} = \begin{bmatrix}
A_{1, k +1} & A_{1, k + 2} & \cdots & A_{1, 2k}\\
A_{2, k +1} & A_{2, k + 2} & \cdots & A_{2, 2k}\\
\vdots & \vdots & \ddots & \vdots\\
A_{k, k +1} & A_{k, k +2} & \cdots & A_{k, 2k}
\end{bmatrix}.
\end{equation}
We then write the top-left $\tilde{n}\times\tilde{n}$ submatrix of $A_n$ in the block-matrix form as 
\begin{equation}\label{eq:blockMatrix}
\begin{bmatrix}
\BB_{11} & \BB_{12} & \cdots & \BB_{1,n'}\\
\BB_{21} & \BB_{22} & \cdots & \BB_{2,n'}\\ 
\vdots & \vdots & \ddots & \vdots\\
\BB_{n',1} & \BB_{n',2} & \cdots & \BB_{n',n'} 
\end{bmatrix}.
\end{equation}
Denote
\begin{align}\label{eq:ButDefn}
\BB_{\mathrm{ut}} \defeq \BB_{12},\BB_{13},\BB_{23},\ldots,\BB_{1,n'},\cdots,\BB_{n'-1,n'}
\end{align}
as the sequence of off-diagonal submatrices in the upper triangle and  
\begin{align}\label{eq:BdiagDefn}
\BB_{\mathrm{d}} \defeq \BB_{11},\BB_{22},\ldots, \BB_{n',n'}
\end{align}
as the sequence of diagonal submatrices.

\item {\bf Binary to $m$-ary conversion.} Let $m := 2^{k^2}$. Each $k \times k$ submatrix with binary entries in the two submatrix sequences $\BB_{\mathrm{ut}}$ and $\BB_{\mathrm{d}}$ is converted into a symbol in $[m]$.

\item {\bf KT probability assignment.} Apply KT sequential probability assignment for the two $m$-ary sequences $\BB_\mathrm{ut}$ and $\BB_\mathrm{d}$ respectively. Given an $m$-ary sequence $x_1,x_2,\ldots,x_N$, \emph{KT sequential probability assignment} defines $N$ conditional probability distributions over $[m]$ as follows. For $j = 0,1,2,\ldots,N-1$, assign conditional probability
\begin{align}
 q_{\mathrm{KT}}(i|x^j) &\defeq q_{\mathrm{KT}}(X_{j+1} = i| X^j = x^j)\nonumber\\ &= \frac{N_i(x^j) + 1/2 }{j + m/2} \quad \text{ for each } i \in [m],\label{eq:KTProbAssgnSeq}
\end{align}
where $X^j \defeq (X_1,\ldots,X_j), x^j \defeq (x_1,x_2,\ldots,x_j)$, and $N_i(x^j) \defeq \sum_{k=1}^j \indi\{x_k = i\}$ counts the number of symbol $i$ in $x^j$.

\item {\bf Adaptive arithmetic coding.} With the KT sequential probability assignments, compress the two sequences $\BB_\mathrm{ut}$ and $\BB_\mathrm{d}$ separately using adaptive arithmetic coding~\cite{Marpe--Schwarz--Wiegand2003} (see description in Algorithm~\ref{alg:arithmetic}).  In case $k = 1$, the diagonal sequence $\BB_\mathrm{d}$ becomes an all-zero sequence since we assume the graph is simple. So we will only compress the off-diagonal sequence $\BB_\mathrm{ut}$.

\begin{algorithm}
\SetKwInOut{Input}{Input}\SetKwInOut{Output}{Output}
\Input{Data sequence $x^N$, alphabet size $m$}
\BlankLine
Initialize $\texttt{lower}=0, \texttt{upper}=1, \texttt{logprob}=0, N_1 = N_2 = \cdots = N_m = 0$\;
\For{$j = 0,1,\ldots,N-1$}{
$\texttt{range} \leftarrow \texttt{upper}-\texttt{lower}$\;
\For{$i = 1,2,\ldots,x_{j+1}$}{
Compute $q_\mathrm{KT}(i|x^j) = \frac{N_i+1/2}{j+m/2}$\;
}
$\texttt{upper} \leftarrow \texttt{lower} + \texttt{range}\cdot\sum_{i=1}^{x_{j+1}}q_{\mathrm{KT}}(i|x^j)$\;
$\texttt{lower} \leftarrow \texttt{upper} - \texttt{range}\cdot q_{\mathrm{KT}}(x_{j+1}|x^j)$\;
$N_{x_{j+1}} \leftarrow N_{x_{j+1}} + 1$\;
$\texttt{logprob} \leftarrow \texttt{logprob} + \log(q_{\mathrm{KT}}(x_{j+1}|x^j))$\;
}
\BlankLine
\Output{the binary representation of $\frac{1}{2}(\texttt{lower}+\texttt{upper})$ with $\lceil-\texttt{logprob}\rceil + 1$ bits}
 \caption{$m$-ary adaptive arithmetic encoding with KT probability assignment}\label{alg:arithmetic}
\end{algorithm}

\item {\bf Encoding the remaining bits.}  The above process compressed the top-left $\tilde{n}\times\tilde{n}$ block of the adjacency matrix. For the remaining $(n-\tilde{n})\tilde{n} + \binom{n-\tilde{n}}{2}$ entries in the upper diagonal of the adjacency matrix, we simply use $2\lceil\log n\rceil$ bits to encode the row and column number of each one.
\end{itemize}

Given the compressed graph sequence $y^L$, the number of vertices $n$ and the submatrix size $k$, the graph decompressor $D_k\suchthat \{0,1\}^* \to \Ac_n$ is defined as follows.
\begin{itemize}
\item {\bf Adaptive arithmetic decoding.} With the KT sequential probability assignments defined in~\eqref{eq:KTProbAssgnSeq}, decompress the two code sequences for $\BB_\mathrm{ut}$ and $\BB_\mathrm{d}$ separately using adaptive arithmetic decoding (see Algorithm~\ref{alg:arithmeticDec}). The length of data sequence $\BB_\mathrm{ut}$ and $\BB_\mathrm{d}$ are $n'(n'-1)/2$ and $n'$ respectively. 
\begin{algorithm}
\SetKwInOut{Input}{Input}\SetKwInOut{Output}{Output}
\Input{Binary sequence $y^L$, alphabet size $m=2^{k^2}$, length of data sequence $N$}
\BlankLine
Add `$0.$' before sequence $y^L$ and convert it into a decimal real number $Y$. Initialize $\texttt{lower}=0, \texttt{upper}=1, N_1 = N_2 = \cdots = N_m = 0$\;
\For{$j = 0,1,\ldots,N-1$}{
$\texttt{range} \leftarrow \texttt{upper}-\texttt{lower}$\;
\For{$i = 1,2,\ldots,m$}{
Compute $q_\mathrm{KT}(i|x^j) = \frac{N_i+1/2}{j+m/2}$\;
}
Find minimum $z\in[m]$ such that $\texttt{lower}+\texttt{range}\cdot \sum_{i=1}^{z} q_\mathrm{KT}(i|x^j)>Y $\;
$\texttt{upper} \leftarrow \texttt{lower} + \texttt{range}\cdot\sum_{i=1}^{z}q_{\mathrm{KT}}(i|x^j)$\;
$\texttt{lower} \leftarrow \texttt{upper} - \texttt{range}\cdot q_{\mathrm{KT}}(z|x^j)$\;
$N_{z} \leftarrow N_{z} + 1$\;
$x_{j+1} \leftarrow z$\;

}

\BlankLine
\Output{the $m$-ary data sequence $x_1,x_2,\cdots,x_N$}
 \caption{$m$-ary adaptive arithmetic decoding with KT probability assignment}\label{alg:arithmeticDec}
\end{algorithm}

 \item {\bf $m$-ary to binary conversion.} Each $m$-ary symbol in the sequence is converted to a $k^2$-bit binary number and further converted into a $k\times k$ submatrix with binary entries. 
 \item {\bf Adjacency matrix recovery.} With the submatrices in $\BB_\mathrm{ut}$ and $\BB_\mathrm{d}$, recover the top-left $\tilde{n}\times\tilde{n}$ submatrix of $A_n$ in the order described in~\eqref{eq:blockMatrix},~\eqref{eq:ButDefn}, and~\eqref{eq:BdiagDefn}.
 \item {\bf Decoding the remaining bits.} Recover the remaining $(n-\tilde{n})\tilde{n} + \binom{n-\tilde{n}}{2}$ entries in the $A_n$ using the row and column numbers of the ones.
\end{itemize}

\begin{rem}[Complexity]
The computational complexity and the space complexity of the proposed compressor are $O(\frac{2^{k^2}n^2}{k^2})$ 
and $O(n^2)$ respectively. 
For the choice of $k$ that achieves universality over $\mathscr{P}_1(\eps)$ family in Theorem~\ref{thm:strongUniv}, $O(\frac{2^{k^2}n^2}{k^2}) = O(\frac{n^{2+\delta}}{k^2})$ for $\delta < \epsilon$. For the choice of $k$ that achieves universality over $\mathscr{P}_2(\eps)$ family in Theorem~\ref{thm:weakUniv}, $O(\frac{2^{k^2}n^2}{k^2}) = O(n^2)$.

   To see this, we consider the time and space complexity of the encoder step by step. The steps of submatrix decomposition and binary to $m$-ary conversion require reading the whole adjacency matrix and computing $m$-ary value of each block. They require $O(n^2)$ time complexity and $O(\frac{n^2}{k^2})$ space complexity.
    Regarding the complexity of adaptive arithmetic coding with KT probability assignment, we consider the ideal case when we are given a computer with infinite floating point precision. In this case, the encoder runs for $O(\frac{n^2}{k^2})$ iterations, and in each iteration, it requires $O(2^{k^2})$ time complexity for computing the conditional probabilities. So the total time complexity is $O(\frac{2^{k^2}n^2}{k^2})$. The space complexity is bounded by the total number of output bits, which is $O(n^2)$. In practice, a computer cannot have infinite floating point precision, and this is a known difficulty in implementing an arithmetic encoder. In~\cite{witten1987}, the authors proposed a practical method of implementing arithmetic coding with finite floating point precision by outputting the known most significant bits of the code during the encoding process. We refer interested readers to their original paper for more details. Finally, the encoding of the remaining bits requires at most $O(nk)$ time complexity and $O(1)$ space complexity. Overall, the encoder requires $O(\frac{2^{k^2}n^2}{k^2})$ time complexity and $O(n^2)$ space complexity. The time and space complexity of the decoder can be analyzed similarly. 
\end{rem}
\begin{rem}
One can check that $C_k$ is well-defined. The submatrix decomposition and the binary to $m$-ary conversion are clearly one-to-one. It is also known that for any valid probability assignment, arithmetic coding produces a prefix code, which is also one-to-one. 
\end{rem}
\begin{rem}
The orders in $\BB_\mathrm{ut}$ and $\BB_\mathrm{d}$ do not matter in terms of establishing universality. The current orders in~\eqref{eq:ButDefn} and~\eqref{eq:BdiagDefn} together with arithmetic coding enable a \emph{horizon-free} implementation. That is, the encoder does not need to know the \emph{horizon} $n$ to start processing the data and can output partially coded bits \emph{on the fly} before receiving all the data. This leads to short encoding and decoding delays. For some real-world applications, for example, when the number of users increases in a large social network, this compressor has the advantage of not requiring to re-process existing data and re-compress the whole graph from scratch.
\end{rem}

\begin{rem}[{\bf Laplace probability assignment}]\label{rmk:Laplace}
 As an alternative to the KT sequential probability assignment, one can also use the Laplace sequential probability assignment. Given an $m$-ary sequence $x_1,x_2,\ldots,x_N$, \emph{Laplace sequential probability assignment} defines $N$ conditional probability distributions over $[m]$ as follows. For $j = 0,1,2,\ldots,N-1$, we assign conditional probability
\begin{equation}\label{eq:LaplaceProbAssgnSeq}
 q_{\Lap}(X_{j+1} = i| X^j = x^j) = \frac{N_i(x^j) + 1 }{j + m} \quad \text{ for each } i \in [m].
\end{equation}
Both methods can be shown to be universal, while Laplace probability assignment has a much cleaner derivation. However, KT probability assignment produces a better empirical performance. For this reason, we keep both in the paper.
\end{rem}

\subsection{Theoretical performance}\label{sec:maintheorems}
We now state the main result of this paper: the proposed compressor $C_k$, for a carefully chosen $k$, is universal over the classes  $\mathscr{P}_1(\eps)$ and  $\mathscr{P}_2(\eps)$ respectively for every $0 < \eps < 1$. 

\begin{thm}[Universality over $\mathscr{P}_1$]\label{thm:strongUniv}
For every $0 < \eps < 1$, the graph compressor $C_k$  is universal over the family $\mathscr{P}_1(\eps)$ provided that
 \[
 0 < \delta < \e, \quad k \le \sqrt{\d \log n}, \quad \text{ and } \quad k = \omega(1).
 \]
\end{thm}

\begin{rem}
Recall that $\mathscr{P}_1(\eps)$ is the family of SBMs with edge probability in the regime $\Omega(\frac{1}{n^{2-\eps}})$ and $O(1)$. Moreover, $\frac{1}{n^2}$ is the threshold for a random graph to contain an edge with high probability~\cite{Frieze--Karonski2015}. Thus, the family $\mathscr{P}_1(\e)$ covers most non-trivial SBM graphs.
\end{rem}

It turns out that for $k=1$, the compressor $C_1$ is universal over the family $\mathscr{P}_2(\eps)$, which we state formally below.

\begin{thm}[Universality over $\mathscr{P}_2$]\label{thm:weakUniv}
For every $0 < \eps < 1$, the graph compressor $C_1$ is universal over the family $\mathscr{P}_2(\eps)$.
\end{thm}

\begin{rem}
Any compressor universal over the class $ \mathscr{P}_1(\eps)$ is also universal over the class  $\mathscr{P}_2(\eps)$, but our compressor designed specifically for the class  $\mathscr{P}_2(\eps)$ has a lower computational complexity. 
\end{rem}

The above results establish the universality of the proposed compressor over $\mathscr{P}_1$ and $\mathscr{P}_2$. The minimax redundancy represents the lowest achievable redundancy of any lossless compressor for the family of SBMs. We bound the minimax redundancy of the family of SBMs $\mathscr{P}_3$ with fixed regime $f(n)$ and fixed largest entry $Q_\mathrm{max}$ in $\mathbf{Q}$ in the following Theorem.
\begin{restatable}[Minimax redundancy]{thm}{Minimax}
\label{thm:minimaxredun}
Let $f(n) = o(1)$ and $f(n) = \Omega(1/n^{2-\eps})$ for some $0 < \eps < 1$. 
The minimax redundancy of the family $\mathscr{P}_3(f(n),Q_\mathrm{max})$ is bounded as
\begin{align*}
    \frac12\log\left(\tfrac{\binom{n}{2}f(n)Q_\mathrm{max}}{\pi e}\right)&\le{R}^*(\mathscr{P}_3(f(n),Q_\mathrm{max}))\\
    &\le \frac{1}{e}(\log e)Q_\mathrm{max}\binom{n}{2}f(n).
\end{align*}
\end{restatable}

\begin{rem}
From Theorem~\ref{thm:minimaxredun}, we can see that the lower bound scales as $\log(n^2f(n))$, while the upper bound scales as $n^2f(n)$. Closing the gap between the upper and lower bounds is an interesting future research question. Recall that the minimax redundancy is the optimal second-order term in the expected length of a universal compressor---we will show in Lemma~\ref{lem:graphEntropy} that $H(A_n)$, the first order term in the length, is $\Theta(n^2 f(n) \log(1/f(n))$. 
\end{rem}

In Section~\ref{sec:BCentropy}, we analyze the performance of the proposed compressor under the local weak convergence framework considered in~\cite{Delgosha--Anatharam2019tit}. It turns out that our proposed compressor achieves the optimal compression rate under the stronger universality criterion considered in that framework. This result requires a lot more definitions to introduce and we defer it to Theorem~\ref{thm:Achieve_BC} in Section~\ref{sec:BCentropy}.

\subsection{Empirical performance}
\label{sec:experiments}
We implement the proposed universal graph compressor (UGC) in four widely used benchmark graph datasets: protein-to-protein interaction network (PPI)~\cite{Grover--Leskovec2016}, LiveJournal friendship network (Blogcatalog)~\cite{Tang--Liu2009}, Flickr user network (Flickr)~\cite{Tang--Liu2009}, and YouTube user network (YouTube)~\cite{Nandanwar--Murty2016}. The submatrix decomposition size $k$ is chosen to be $1,2,3,4$ and we present in Table~\ref{tab:UGC} the compression ratios (the ratio between output length and input length of the encoder) of UGC for different choices of $k$. We present in Table~\ref{tab:simulation} the compression ratios of four competing algorithms\footnote{Note that CSR and Ligra+ are designed to enable fast computation, such as adjacency query or vertex degree query, in addition to compressing the matrix. Our proposed compressor does not possess such functionality and is designed solely for compression purposes.}. In Fig.~\ref{fig:log-scale-simulation}, we combine the comparisons of the compression ratios in Tables~\ref{tab:UGC}~and~\ref{tab:simulation} in the logarithmic scale.

\begin{figure*}[htbp]
  \centering
    \includegraphics[scale=.45]{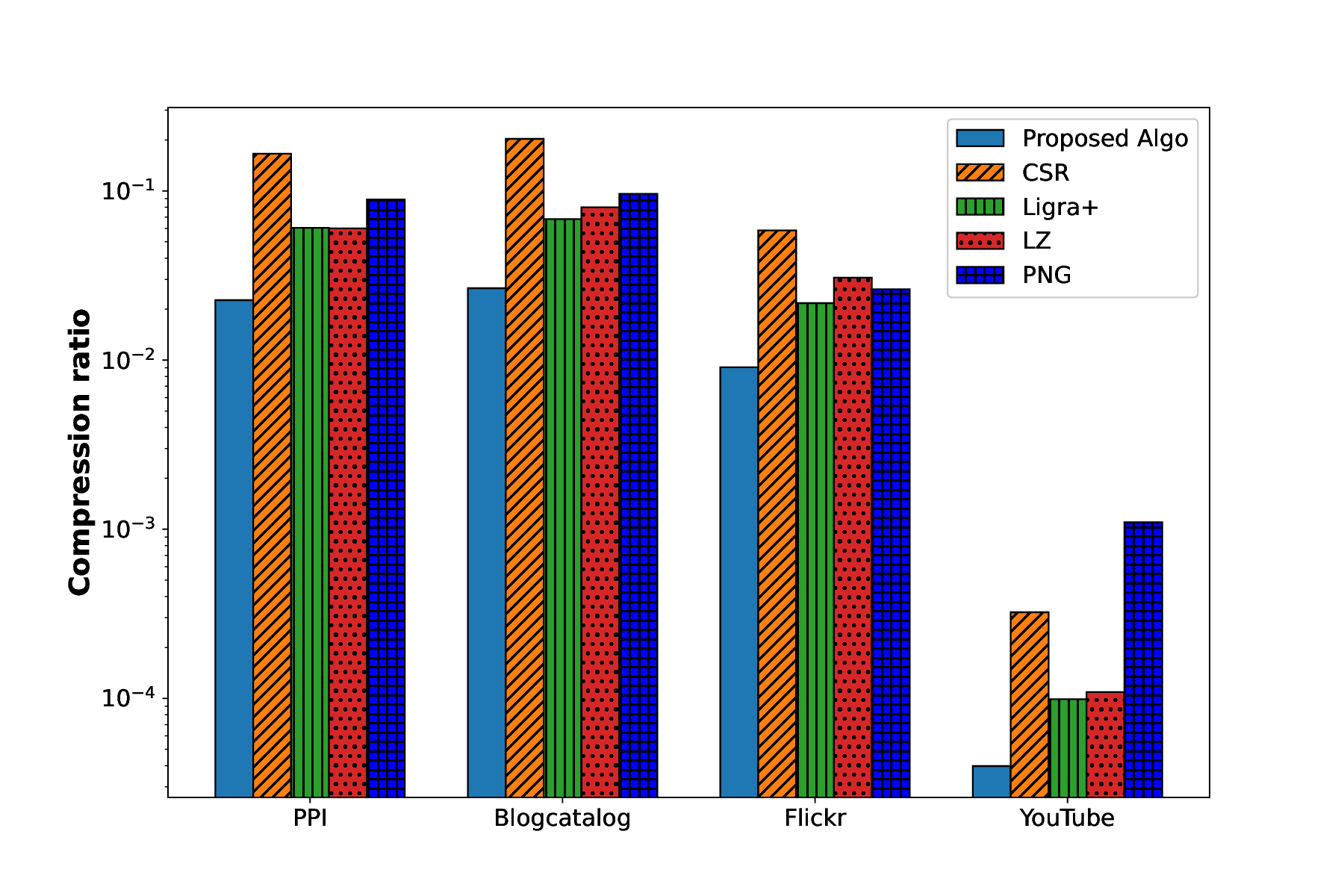}
    \caption{Log-scale comparisons of the compression ratios for the proposed universal graph compressor with other competing compressors. See Tables~\ref{tab:UGC} and~\ref{tab:simulation} for the exact compression ratios. The column numbers of highlighted entries in Table~\ref{tab:UGC} indicate the optimal $k$ value for UGC we used in this plot.  }
    \label{fig:log-scale-simulation}
\end{figure*}

\begin{itemize}
 \item CSR: Compressed sparse row is a widely used sparse matrix representation format. In the experiment, we further optimize its default compressor exploiting the fact that the graph is simple and its adjacency matrix is symmetric with binary entries. 
  \item Ligra+: This is another powerful sparse matrix representation format~\cite{Shun--Blelloch2013,Shun--Dhulipala--Blelloch2015}, which improves upon CSR using byte codes with run-length coding.
 \item LZ: This is an implementation of the algorithm proposed in~\cite{Lempel--Ziv86}, which first transforms the two-dimensional adjacency matrix into a one-dimensional sequence using the Peano--Hilbert space-filling curve and then compresses the sequence using Lempel--Ziv 78 algorithm~\cite{Lempel--Ziv78}.
 \item PNG: The adjacency matrix of the graph is treated as a grey-scaled image and the PNG lossless image compressor is applied.
\end{itemize}
The compression ratios of the five algorithms implemented on four datasets are given as follows. The proposed UGC outperforms all competing algorithms in all datasets. The compression ratios from competing algorithms are 2.4 to 27 times that of the universal graph compressor. 

\begin{table*}[ht]
\caption{\\Compression ratio of UGC under different $k$ values.}\label{tab:UGC}
\centering
            \begin{tabular}{||c |c | c | c | c ||} 
            \hline
            & $k=1$  & $k=2$ & $k=3$ & $k=4$  \\ 
            \hline
            PPI & 0.0229 & $\mathbf{0.0226}$ & 0.0227 & 0.0343 \\
            \hline
            Blogcatalog & 0.0275 & 0.0270 & $\mathbf{0.0266}$ & 0.0287\\
            \hline
            Flickr & 0.00960 & 0.00934 & 0.00914 & $\mathbf{0.00906}$ \\
            \hline
            YouTube & $4.51 \times 10^{-5}$ & $4.11 \times 10^{-5}$ & $\mathbf{3.98 \times 10^{-5}}$ & $4.00 \times 10^{-5}$\\
            \hline
            \end{tabular}

\end{table*}
\begin{table*}[ht]
\caption{\\Compression ratios of competing algorithms.}\label{tab:simulation}
\centering
            \begin{tabular}{||c  | c | c | c | c||} 
            \hline
            & CSR & Ligra+ & Lempel--Ziv & PNG  \\ 
            \hline
            PPI  & 0.166 & 0.0605 & {0.06} & 0.089  \\
            \hline
            Blogcatalog & 0.203 & {0.0682} & 0.080 & 0.096 \\
            \hline
            Flickr & 0.0584 & 0.0217 & 0.0307 & 0.0262 \\
            \hline
            YouTube & $3.23\times 10^{-4}$ & $9.90\times 10^{-5}$& $1.09\times 10^{-4}$ & $1.10\times 10^{-3}$ \\
            \hline
            \end{tabular}

\end{table*}

\section{Main Ideas in Establishing Universality}\label{sec:theorem}
In this section, we establish Theorems~\ref{thm:strongUniv} and~\ref{thm:weakUniv} in Section~\ref{sec:maintheorems} based on four intermediate propositions. We defer the proofs of the propositions to Section~\ref{sec:proposition}.
\subsection{Graph Entropy}
We first calculate the entropy of the (random) graph $A_n$, which, recall, is the fundamental lower bound on the expected compression length for any compression scheme. 
Since we want to bound the redundancy of a universal compressor, we will be concerned with both the first and the second-order terms in $H(A_n)$. Recall, for $p \in [0,1], h(p)$ denotes the binary entropy function and for a matrix $W$ with entries in $[0,1]$, $h(W)$ denotes a matrix of the same dimension whose $(i,j)$ entry is $h(W_{ij})$.

\begin{restatable}[Graph entropy]{prop}{graphEntropy}
\label{prop:graphEntropy}
 Let $A_n \sim \mathrm{SBM}(n,L,\pv,f(n)\Qv)$ with $f(n) = O(1), f(n) = \Omega\left(\frac{1}{n^2}\right)$, and $L = \Theta(1)$. Recall that $X_i\in[L]$ denotes the community label of vertex $i$. 
 Then
 {\allowdisplaybreaks
 \begin{align}
  \label{eq:graphEntropy}
  H(A_n) &= {n \choose 2} H(A_{12}|X_1,X_2) (1+o(1))\\
  &= {n \choose 2}\pv^T h\bigl(f(n)\Qv\bigr) \pv + o\left(n^2h\bigl(f(n)\bigr)\right).\label{eq:graphEntropy2}
 \end{align}}%
 In particular, when $f(n) = \Omega\left(\frac{1}{n^2}\right)$ and $f(n) = o(1)$, expression~\eqref{eq:graphEntropy2} can be further simplified as
\begin{equation}
 \label{eq:graphEntropyo(1)}
  H(A_n) =  {n \choose 2} f(n)\log\left(\frac{1}{f(n)}\right)(\pv^T\Qv\pv + o(1)).
\end{equation}
\end{restatable}

\begin{rem}
 In the regime $f(n) = \Omega\left(\frac{1}{n}\right)$ and $f(n) = O(1)$, equation~\eqref{eq:graphEntropy2} has been established in~\cite{Abbe16}. We extend the analysis to the regime $f(n) = \Omega(\frac{1}{n^2})$ and $f(n) = o\left(\frac{1}{n}\right)$. 
\end{rem}

\begin{rem}
 Proposition~\ref{prop:graphEntropy} can be used to calculate the entropy of the graph for certain important regimes of $f(n)$, in which the SBM displays characteristic behavior. For $f(n) = 1$, we have $H(A_n) = {n \choose 2}(\pv^T h\left(\Qv \right)\pv + o(1))$; for $f(n) = \frac{\log n}{n}$ (the regime where the phase transition for exact recovery of the community assignments occurs~\cite{Abbe--Bandeira--Hall15, Abbe--Sandon15}), we have $H(A_n) = \frac{ n \log^2  n}{2}(\pv^T \Qv \pv+ o(1))$; when $f(n) = \frac{1}{n}$ (the regime where the phase transition  for detection between SBM and the \erdos--\renyi model occurs~\cite{Mossel--Neeman--Sly15}), we have $H(A_n) = \frac{n \log n}{2}(\pv^T \Qv \pv  + o(1))$; when $f(n) = \frac{1}{n^2}$ (the regime where the phase transition for the existence of an edge occurs), we have $H(A_n) = \log n(\pv^T \Qv \pv  + o(1))$.
\end{rem}

\subsection{Asymptotic i.i.d. via Submatrix Decomposition}
To compress the matrix $A_n$, we wish to decompose it into a large number of components that have little correlation between them. This leads to the idea of submatrix decomposition described previously. Since the sequence of submatrices is used to compress $A_n$, the next proposition claims these submatrices are identically distributed and asymptotically independent in a precise sense described as follows.

\begin{restatable}[Submatrix decomposition]{prop}{decomp}
\label{prop:blkIndep}
 Let $A_n \sim \mathrm{SBM}(n,L,\pv,f(n)\Qv)\in\mathscr{P}_1(\eps)$ for some $0<\eps<1$.
Let $k$ be an integer and $n' = \lfloor n/k\rfloor$. Consider the $k\times k$ submatrix decomposition in~\eqref{eq:blockMatrix}. We have all the off-diagonal submatrices share the same joint distribution; all the diagonal submatrices share the same joint distribution. In other words, for any $1 \le i_1,i_2,j_1,j_2 \le n'$ with $i_1 \neq j_1, i_2 \neq j_2$ and $1 \le l_1,l_2 \le n'$, we have 
\begin{align*}
    \BB_{i_1,j_1} &\stackrel{d}{=} \BB_{i_2,j_2}, \\
    \BB_{l_1,l_1} &\stackrel{d}{=} \BB_{l_2,l_2}.
\end{align*}
In addition, if $k = \omega(1)$ and $k = o(n)$, the off-diagonal submatrices are asymptotically i.i.d.\! in the sense that 
\begin{equation}\label{eq:AsympIndepBlks}
    \lim_{n\to \infty}\frac{H(\BB_{\mathrm{ut}})}{{n' \choose 2}H(\BB_{12})} = 1.
\end{equation}
\end{restatable}

\begin{rem}
In Proposition~\ref{prop:blkIndep}, we provide a more general result than what we need to prove Theorem~\ref{thm:strongUniv}.
     In Theorem~\ref{thm:strongUniv}, the range of $k$ for the compressor to achieve universality over the family $\mathscr{P}_1(\epsilon)$ is given by
     \[
 0 < \delta < \e, \quad k \le \sqrt{\d \log n}, \quad \text{ and } \quad k = \omega(1),
 \]
    while in Proposition~\ref{prop:blkIndep}, we show the asymptotic i.i.d property for a wider range  $k=\omega(1)$ and $k=o(n)$. 
\end{rem}

\begin{rem}
Proposition~\ref{prop:blkIndep} implies that any probability assignment that is universal for i.i.d. sequences may be used to compress the blocks, provided that the compression length analysis can be expended to asymptotically i.i.d. sequences as in Propositions~\ref{prop:NonIIDcomp} and~\ref{prop:KTlength}. The reason we choose to use the KT/Laplace probability assignment is that it is known to achieve the minimax redundancy and lead to better finite-length performance as compared to, for example, Lempel--Ziv coding which while still being universal incurs larger redundancy. Moreover, the KT-based approach leads to better empirical performance, as observed in Table~\ref{tab:simulation}.
\end{rem}
\subsection{Length Analysis for Correlated Sequences}
Thanks to the asymptotic i.i.d. property of the submatrix decomposition, we hope to compress these submatrices as if they are independent using a Laplace probability assignment (which, recall, is universal for the class of all $m$-ary i.i.d. processes). However, since these submatrices are still correlated (albeit weakly), we will need a result on the performance of Laplace probability assignment on correlated sequences with identical marginals, which we give next.

\begin{restatable}[Laplace probability assignment for correlated sequence]{prop}{laplace}
\label{prop:NonIIDcomp}
Consider arbitrarily correlated $Z_1, Z_2,$ $\ldots, Z_N$, where each $Z_i$ is identically distributed over an alphabet of size $m \geq 2$. Let $\ell_{\Lap}(z^N) = \left\lceil\log \frac{1}{q_{\Lap}(z^N)}\right\rceil+1$, where 
\begin{align}\label{eq:LaplaceProbAssgn}
    q_{\Lap}(z^N) \defeq \frac{N_1!N_2!\cdots N_m!}{N!}\cdot \frac{1}{{N+m-1 \choose m-1}}
\end{align}
with $N_i \triangleq N_i(z^N) = \sum_{j=1}^N \indi\{z_j = i\}$ is the joint distribution induced by the sequential Laplace probability assignment in~\eqref{eq:LaplaceProbAssgnSeq}.
We then have
\begin{equation}
\E[\ell_{\Lap}(Z^N)] \le m\log(2eN) + N H(Z_1)+2. \label{eq:LenLaplaceGen}
\end{equation}
\end{restatable}

We provide a similar result for the KT probability assignment.

\begin{restatable}[KT probability assignment for correlated sequence]{prop}{kt}
\label{prop:KTlength}
 Consider arbitrarily correlated $Z_1, Z_2, \ldots,$ $Z_N$, where each $Z_i$ is identically distributed over an alphabet of size $m \geq 2$. Let $\ell_{\mathrm{KT}} (z^N)=\left\lceil\log \frac{1}{q_{\mathrm{KT}}(z^N)}\right\rceil+1$, where 
\begin{align}\label{eq:KTProbAssgn}
    q_{\mathrm{KT}}(z^N) = \frac{(2N_1 -1)!!(2N_2 -1)!!\cdots (2N_m -1)!!}{m(m+2)\cdots (m+2N-2)}
\end{align}
with $(-1)!! \triangleq 1$ and $N_i \triangleq N_i(z^N) = \sum_{j=1}^N \indi\{z_j = i\}$ is the joint distribution induced by KT probability assignment in~\eqref{eq:KTProbAssgnSeq}.
We then have
  \begin{equation}
\E[\ell_{\mathrm{KT}}(Z^N)]\le \tfrac{m}{2}\log\left(e\bigl(1+\tfrac{2N}{m}\bigr)\right)+\tfrac 12\log(\pi N)+NH(Z_1)+2.\label{eq:LenKT}
\end{equation}
\end{restatable}

\subsection{Proof of Theorem~\ref{thm:strongUniv}}
To prove Theorem~\ref{thm:strongUniv}, we first state a technical lemma, the proof of which is deferred to Section~\ref{sec:graphEntropy}.
\begin{lem}\label{lem:technical}
    Let $A_n \sim \mathrm{SBM}(n,L,\pv,f(n)\Qv)$. Recall that in the stochastic block model, $X_i$ denotes the community label of vertex $i$. 
    We have
    \begin{equation}
    \label{eq:A12-entropy}
        H(A_{12})=h(f(n)\PQP), 
    \end{equation}
    and 
    \begin{equation}
    \label{eq:A12-cond-entropy}
        H(A_{12}|X_1,X_2)=\pv^T h\bigl(f(n)\Qv\bigr)\pv.
    \end{equation}
    Moreover, if $f(n)=O(1)$, then 
    \begin{equation}
    \label{eq:entropy-theta}
        h(f(n)\PQP)=\Theta(\pv^T h\bigl(f(n)\Qv\bigr)\pv), 
    \end{equation}
    and if $f(n)=o(1)$, we further have 
    \begin{equation}
    \label{eq:entropy-1+o1}
        h(f(n)\PQP)=(1+o(1))\pv^T h\bigl(f(n)\Qv\bigr)\pv.
    \end{equation}
\end{lem}
We are now ready to prove Theorem~\ref{thm:strongUniv}.

\begin{proof}[{\bf Proof of Theorem~\ref{thm:strongUniv}}]
We will prove the universality of $C_k$ for both the KT probability assignment and the Laplace probability assignment. Note that the upper bound on the expected length of KT in~\eqref{eq:LenKT} is upper bounded by the upper bound on the length of Laplace in~\eqref{eq:LenLaplaceGen}. So it suffices to show Laplace probability assignment is universal.

 We use the bound in Proposition~\ref{prop:NonIIDcomp} to establish the upper bound on the length of the code. Recall that here we compress the diagonal submatrices $\BB_{\mathrm{d}}$ ($m = 2^{k^2}-$sized alphabet, $N = n'$ submatrices) and the off-diagonal submatrices $\BB_{\mathrm{ut}}$  ($m = 2^{k^2}-$sized alphabet, $N = {n' \choose 2}$ submatrices) separately, and encode the position of each $1$ in the remaining $(n-\tilde{n})\tilde{n} + \binom{n-\tilde{n}}{2}$ entries using $2\lceil \log n \rceil$ bits. Let $N_r$ denote the number of ones in the remaining $(n-\tilde{n})\tilde{n} + \binom{n-\tilde{n}}{2}$ entries. We have, 
 
 \begin{align}
  &\frac{\E(\ell(C_k(A_n)))}{H(A_n)}\nonumber\\ &= \frac{\E(\ell_{\Lap}(\BB_{\mathrm{ut}})) + \E(\ell_{\Lap}(\BB_{\mathrm{d}}))+\E(2N_r\lceil\log n\rceil)}{H(A_n)}\nonumber\\
   &\le \frac{{n' \choose 2}H(\BB_{12}) + 2^{k^2}\log\left(2e{n' \choose 2}\right) + n'H(\BB_{11})}{H(A_n)} \nonumber\\&\;\;\;+ \frac{2^{k^2}\log(2en')+4+\E(2N_r\lceil\log n\rceil)}{H(A_n)}\nonumber \\
   &\stackrel{(a)}{\le}  \frac{{n' \choose 2}H(\BB_{12}) + 2^{k^2}\log\left(en^2\right) + n H(\BB_{11}) }{H(A_n)} \nonumber\\&\;\;\;+ \frac{2^{k^2}\log(2en)+4+\E(2N_r\lceil\log n\rceil)}{H(A_n)}\nonumber \\
   &\stackrel{(b)}{\le} \frac{{n' \choose 2}H(\BB_{12}) + 2^{k^2}\log\left(2e^2n^3\right) + nk^2 H(A_{12})}{H(A_n)} \nonumber\\&\;\;\;+ \frac{4+\E(2N_r\lceil\log n\rceil)}{H(A_n)}\nonumber \\ 
   &= \frac{{n' \choose 2}H(\BB_{12})}{H(A_n)} + \frac{2^{k^2}\log\left(2e^2n^3\right)+4}{H(A_n)} \nonumber\\&\;\;\;+ \frac{nk^2 H(A_{12})}{H(A_n)}+\frac{\E(2N_r\lceil\log n\rceil)}{H(A_n)}\label{eqn:compress-length},
 \end{align}%
 where in (a) we bound ${n' \choose 2} \le n^2$ and $n' \le n$, and in (b) we note that $H(\BB_{11}) \le k^2 H(A_{12})$ since there are $k^2 - k$ elements of the matrix (all apart from the diagonal elements) are distributed identically as $A_{12}$. We will now  analyze each of these four terms separately. Firstly, using Proposition~\ref{prop:blkIndep} yields that $\frac{{n' \choose 2} H(\BB_{12})}{H(A_n)} \to 1$. Next, since $f(n) = \Omega\left(\frac{1}{n^{2-\e}}\right)$, we have $H(A_n) = \Omega(n^{\e} \log n)$ and subsequently substituting $k \le \sqrt{\d \log n}$, we have  
 \begin{equation}
 \frac{2^{k^2} \log (2en^3)+4}{H(A_n)} =  O\left(\frac{n^{\d} \log n}{n^{\e} \log n}\right)  = O\left(n^{\d - \e}\right) = o(1)\label{eqn:small-term-1}
 \end{equation}
 since $\d < \e$. Moreover, we have 
  \begin{align}
 \frac{n k^2 H(A_{12})}{H(A_n)} &\le \frac{n k^2 H(A_{12})}{H(A_n|X^n)}\nonumber\\ & = \frac{n k^2 H(A_{12})}{{n \choose 2}H(A_{12}|X_1,X_2)} \nonumber\\ &= O\left(\frac{k^2}{n}\right) = o(1),\label{eqn:small-term-2}
 \end{align}
 where the penultimate equality used the fact that $H(A_{12}) \sim H(A_{12}|X_1,X_2)$ by Lemma~\ref{lem:technical}.

 For the last term in~\eqref{eqn:compress-length}, we have
  \begin{align}
     \frac{\E(2N_r\lceil\log n\rceil)}{H(A_n)}&=\frac{f(n)\PQP\left((n-\tilde{n})\tilde{n} + \binom{n-\tilde{n}}{2}\right)2\lceil\log n\rceil}{H(A_n)}\nonumber \\ &\stackrel{(c)}{=}O\left(\frac{knf(n)\log n}{\binom{n}{2}f(n)\log \frac{1}{f(n)}}\right)=o(1),\label{eqn:small-term-3}
 \end{align}
 where (c) follows because $n-\tilde{n}=n-n\lfloor n/k\rfloor <k$, $\tilde{n}\le n$ and $k=O(\sqrt{\log n})$.
 We have then established that 
 \begin{align*}
  \frac{\E(\ell(C_k(A_n)))}{H(A_n)} &\le \frac{{n' \choose 2}H(\BB_{12})}{H(A_n)} + \frac{2^{k^2}\log\left(2en^3\right)+4}{H(A_n)}\\  &+ \frac{nk^2 H(A_{12})}{H(A_n)}+\frac{\E(2N_r\lceil\log n\rceil)}{H(A_n)} \\
  &= 1 + o(1),
 \end{align*}
 which finishes the proof.
\end{proof}
\subsection{Proof of Theorem~\ref{thm:weakUniv}}
\begin{proof}
To establish the universality of $C_1$ over the family $\mathscr{P}_2$, we follow a similar argument as in the proof of Theorem~\ref{thm:strongUniv}. In the case when $k=1$, we do not need to compress the diagonal submatrix sequence since they are all zeros and we do not have any remaining bits to compress separately. By Proposition~\ref{prop:NonIIDcomp} with $N = \binom{n}{2}$, we have
\begin{align*}
    &\frac{\E(\ell(C_1(A_n)))}{H(A_n)}\\ &\le  \frac{2\log(2eN) + N H\left(A_{12}\right)+2}{H(A_n)} \\
    &\le \frac{2\log(2eN) + N H\left(A_{12}\right)+2}{H(A_n|X_1^n)} \\
    &= \left(\frac{2\log(2eN) + N H\left(A_{12}\right)+2}{N H(A_{12})} \right) \frac{H(A_{12})}{H(A_{12}|X_1,X_2)} \\
    &= \left(1 + \frac{2\log(2eN)+2}{Nh(f(n)\pv^T \Qv \pv)}  \right)\frac{h(f(n)\pv^T \Qv \pv)}{\pv^T h(f(n)\Qv) \pv}\\
    &\stackrel{(a)}{=} 1 + o(1).
\end{align*}
Here, (a) is justified by noting that $\frac{2\log(2eN)+2}{Nh(f(n)\pv^T \Qv \pv)} \le \frac{2\log(2en^2)+2}{{n \choose 2}h(n^{-(2-\e)}\pv^T \Qv \pv)} \frac{h(n^{-(2-\e)}\pv^T \Qv \pv)}{h(f(n)\pv^T \Qv \pv)}$, and then noting that  $\frac{2\log(2en^2)+2}{{n \choose 2}h(n^{-(2-\e)}\pv^T \Qv \pv)} = o(1)$ and $\frac{h(n^{-(2-\e)}\pv^T \Qv \pv)}{h(f(n)\pv^T \Qv \pv)} = O(1)$ when $f(n) = \Omega\left(\frac{1}{n^{2-\e}}\right)$ and that $h(f(n)\pv^T \Qv \pv) =(1+o(1)) \pv^T h(f(n)\Qv) \pv$ when $f(n)=o(1)$ by Lemma~\ref{lem:technical}. 
\end{proof}

\begin{rem}
When $f(n) = 1$, the compressor $C_1$ is strictly suboptimal. This is because the length achieved by $C_1$ is ${n \choose 2}h\left(f(n) \pv^T \Qv \pv \right) (1 + o(1))$, whereas the first order term in the entropy is ${n \choose 2} \pv^T  h\left(f(n)\Qv\right)  \pv$. When $f(n)$ is $o(1)$, these two have the same first-order term. However, when $f(n)$ is constant, $\pv^T  h\left(f(n)\Qv\right)  \pv$ is strictly smaller than $h\left(f(n) \pv^T \Qv \pv \right) $ by concavity of entropy.
\end{rem}
\section{Proof of Propositions~\ref{prop:graphEntropy}--\ref{prop:KTlength}}\label{sec:proposition}
\subsection{Graph Entropy}
\label{sec:graphEntropy}
In this section, we establish Proposition~\ref{prop:graphEntropy}. Towards that goal, we state and prove a stronger result on the graph entropy in Lemma~\ref{lem:graphEntropy}, which will also be useful for proof of Theorem~\ref{thm:minimaxredun} later. Proposition~\ref{prop:graphEntropy} will follow easily from Lemma~\ref{lem:graphEntropy}. The proofs of Lemmas~\ref{lem:technical} and~\ref{lem:graphEntropy} are provided after the proof of Proposition~\ref{prop:graphEntropy}.
\begin{lem}[Graph entropy]\label{lem:graphEntropy}
 Let $A_n \sim \mathrm{SBM}(n,L,\pv,f(n)\Qv)$ with $f(n) = O(1), f(n) = \Omega\left(\frac{1}{n^2}\right)$, and $L = \Theta(1)$. 
 Then
 {\allowdisplaybreaks
 \begin{align}
  \label{eq:graphEntropyl}
  H(A_n) &= {n \choose 2} H(A_{12}|X_1,X_2) (1+o(1))\\
  &= {n \choose 2}\pv^T h\bigl(f(n)\Qv\bigr) \pv + o\left(n^2h\bigl(f(n)\bigr)\right).\label{eq:graphEntropyl2}
 \end{align}}%

 In particular, when $f(n) = \omega\left(\frac{1}{n}\right)$ and $f(n) = o(1)$, expression~\eqref{eq:graphEntropyl2} can be simplified as\footnote{In the concurrent work~\cite{Delgosha--Anantharam2021b}, the authors derive the first and second order terms in the entropy of sparse graphons. Proposition 2 in~\cite{Delgosha--Anantharam2021b} could specialize to the case of stochastic block models with $f(n) = \omega\left(\frac{1}{n}\right)$ and $f(n) = o(1)$, and yield the same result as in equation~\eqref{eq:Entropy2ndOrder}.}
 \begin{align}
     H(A_n)=&\binom{n}{2}f(n)\bigg(\log\left(\frac{1}{f(n)}\right)\PQP+\PQP\log e\nonumber\\ +&\PsQP+o(1)\bigg),\label{eq:Entropy2ndOrder}
 \end{align}
 where $\mathbf{Q^*}$ denotes an $L\times L$ matrix whose $(i,j)$ entry is $Q_{ij}\log (\frac{1}{Q_{ij}})$ when $Q_{ij}\neq 0$ and $0$ when $Q_{ij}=0$.
 
 When $f(n) = \Omega\left(\frac{1}{n^2}\right)$ and $f(n) = O\left(\frac{1}{n}\right)$, the entropy $H(A_n)$ can be upper bounded as
\begin{align}
  H(A_n)\le &\binom{n}{2}f(n)\bigg(\log\left(\frac{1}{f(n)}\right)\PQP+\PQP\log e\nonumber\\  &+\PQP\log\frac{1}{\PQP}+o(1) \bigg),\label{eq:secondorderupper}
\end{align}
and lower bounded as
\begin{align}
  H(A_n)\ge &\binom{n}{2}f(n)\bigg(\log\left(\frac{1}{f(n)}\right)\PQP+\PQP\log e\nonumber\\ &+\PsQP+o(1) \bigg).\label{eq:secondorderlower}
\end{align}

\end{lem}

Now, we are ready to prove Proposition~\ref{prop:graphEntropy}, which is restated as follows.
\graphEntropy*

\begin{proof}[\bf Proof of Proposition~\ref{prop:graphEntropy}]
Given Lemma~\ref{lem:graphEntropy}, it suffices to prove equation~\eqref{eq:graphEntropyo(1)} in Proposition~\ref{prop:graphEntropy}. We consider two different cases. Firstly, assume that $f(n)=\omega\left(\frac1n\right)$ and $f(n)=o(1)$. Recall that the case of $f(n) = o(1)$ is particularly interesting as the simple compressor $C_1$ attains universality. In this case, by equation~\eqref{eq:Entropy2ndOrder}, we have
\begin{align*}
    H(A_n)&=\binom{n}{2}f(n)\bigg(\log\left(\frac{1}{f(n)}\right)\PQP+\PQP\log e\\ &\;\;\;+\PsQP+o(1)\bigg)\\
    &={n \choose 2} f(n)\log\left(\frac{1}{f(n)}\right)(\pv^T\Qv\pv + o(1)).
\end{align*}
because $\PQP=\Theta(1)$, $\PsQP=\Theta(1)$ and $\log\left(\frac{1}{f(n)}\right)=\omega(1)$. Secondly, assume that $f(n)=O(\frac{1}{n})$ and $f(n)=\Omega(\frac{1}{n^2})$. In this case, the right-hand side of equations~\eqref{eq:secondorderupper} and~\eqref{eq:secondorderlower} can both be written as ${n \choose 2} f(n)\log\left(\frac{1}{f(n)}\right)(\pv^T\Qv\pv + o(1))$ because $\PQP=\Theta(1)$, $\PsQP=\Theta(1)$ and $\log\left(\frac{1}{f(n)}\right)=\omega(1)$. Therefore, equation~\eqref{eq:graphEntropyo(1)} follows.
\end{proof}

Finally, we establish Lemmas~\ref{lem:technical} and~\ref{lem:graphEntropy}.

\begin{proof}[\bf Proof of Lemma~\ref{lem:technical}]
    We note that
    \begin{align*}
        &\P(A_{12}=1)\\&=\sum_{1\le i,j\le L}\P(A_{12}=1,X_1=i,X_2=j)\\
        &=\sum_{1\le i,j\le L}\P(A_{12}=1|X_1=i,X_2=j)P(X_1=i,X_2=j)\\
        &=\sum_{1\le i,j\le L}f(n)Q_{ij}p_ip_j=f(n)\PQP.
    \end{align*}
    Therefore, $A_{12}\sim \bern(f(n)\PQP)$, which implies that $H(A_{12})=h(f(n)\PQP)$.
    For $H(A_{12}|X_1,X_2)$, by the definition of conditional entropy, we have
    \begin{align*}
        H(A_{12}|X_1,X_2)&=\sum_{i,j} H(A_{12}|X_1 = i,X_2 = j)p_i p_j
        \\
        &\stackrel{(a)}{=}\sum_{i,j} h(f(n)Q_{ij})p_i p_j\\
        &=\pv^T h\bigl(f(n)\Qv\bigr)\pv,
    \end{align*}
    where (a) follows because $A_{12}|X_1=i,X_2=j\sim\bern(f(n)Q_{ij})$.
    
    To see that $h(f(n)\PQP)=\Theta(\pv^T h\bigl(f(n)\Qv\bigr)\pv)$ when $f(n)=\Theta(1)$, we notice that $\PQP=\Theta(1)$ because all entries in $\pv$ and $\Qv$ are constant. Moreover, $\PQP<1$ because $\max_{i,j}Q_{ij}<1$ and all entries in $\pv$ are non-zero. As a result, we have $h(f(n)\PQP)=\Theta(1)$. For $\pv^T h\bigl(f(n)\Qv\bigr)\pv$, similarly, we observe that $\max_{i,j}Q_{ij}=\Theta(1)$ and $\max_{i,j}Q_{ij}<1$, so $\pv^T h\bigl(f(n)\Qv\bigr)\pv=\Theta(1)$ since all entries in $\pv$ are constants. This completes the proof of~\eqref{eq:entropy-theta}.

    Now, consider the case of $f(n)=o(1)$. By the property of binary entropy function (see for example~\cite{Courtade2012}), if $f(n)=o(1)$, we can expand $h(f(n)\PQP)$ and $\pv^T h\bigl(f(n)\Qv\bigr)\pv$ as
    \begin{align}
        h(f(n)\PQP)=&f(n)\bigg(\log \left(\frac{1}{f(n)}\right)\PQP+\PQP\log e\nonumber\\ &+\mathbf{p}^T\mathbf{Qp}\log\frac{1}{\PQP}+o(1)\bigg),\label{eq:hpqp-expand}
    \end{align}
    and
    \begin{align}
        \pv^T h\bigl(f(n)\Qv\bigr)\pv= &f(n)\bigg(\log \left(\frac{1}{f(n)}\right)\PQP+\PQP\log e\nonumber\\ &+\mathbf{p}^T\mathbf{Q^*p}+o(1)\bigg), \label{eq:phqp-expand}
    \end{align}
    where $\mathbf{Q^*}$ denotes an $L\times L$ matrix whose $(i,j)$ entry is $Q_{ij}\log (\frac{1}{Q_{ij}})$ when $Q_{ij}\neq 0$ and $0$ when $Q_{ij}=0$. 
    To see~\eqref{eq:entropy-1+o1}, we take~\eqref{eq:hpqp-expand} minus~\eqref{eq:phqp-expand} and get
    \begin{align*}
        &h(f(n)\PQP)-\pv^T h\bigl(f(n)\Qv\bigr)\pv\\
        &=f(n)\left(\mathbf{p}^T\mathbf{Qp}\log\frac{1}{\PQP}-\mathbf{p}^T\mathbf{Q^*p}+o(1)\right)\\
        &\stackrel{(b)}{=}o(h(f(n)\PQP)),
    \end{align*}
    where (b) follows because $\log \left(\frac{1}{f(n)}\right)=\omega(1)$ and $\mathbf{p}^T\mathbf{Qp}\log\frac{1}{\PQP}=\Theta(1)$ and $\mathbf{p}^T\mathbf{Q^*p}=\Theta(1)$. This completes the proof of~\eqref{eq:entropy-1+o1}.
\end{proof}

\begin{proof}[\bf Proof of Lemma~\ref{lem:graphEntropy}]
 Note that
 {\allowdisplaybreaks
\begin{align}
    H(A_n) &= H(A_n|X^n) + I(X^n;A_n)\nonumber \\
    &= {n \choose 2}H(A_{12}|X_1,X_2) + I(X^n;A_n) \label{eq:condindepAij}\\
    &= {n \choose 2}\pv^T h\bigl(f(n)\Qv\bigr)\pv  + I(X^n;A_n), \label{eq:upperboundEntropyRateCommLabels}
\end{align}}%
where the last equality follows by Lemma~\ref{lem:technical} and the fact that all the ${n \choose 2}$ edges are identically distributed and also independent given $X^n$.
Also, notice that
\begin{align*}
    h(f(n))&=f(n)\log\frac{1}{f(n)}+(1-f(n))\log\frac{1}{1-f(n)}\\
    &=\Theta(1)=\Theta\left(f(n)\log\frac{1}{f(n)}\right)
\end{align*}
when $f(n)=\Theta(1)$, and 
\begin{align*}
    h(f(n))&=f(n)\log\frac{1}{f(n)}+(1-f(n))\log\frac{1}{1-f(n)}\\
    &=f(n)\log\frac{1}{f(n)}(1-f(n))\frac{f(n)}{1-f(n)}\\
    &=f(n)\log\frac{1}{f(n)}(1+o(1))
\end{align*}
when $f(n)=o(1)$. As a result, to prove~\eqref{eq:graphEntropyl2}, if suffices to prove 
\[
H(A_n)=  {n \choose 2}\pv^T h\bigl(f(n)\Qv\bigr) \pv + o\left(n^2f(n)\log\frac{1}{f(n)}\right).
\]
In the following, we first establish statement~\eqref{eq:graphEntropyl2} for $f(n) = \Theta(1)$ and then prove the simplified expression~\eqref{eq:Entropy2ndOrder} for $f(n) = \omega\left(\frac{1}{n}\right)$ and $f(n) = o(1)$. Finally, we prove the refined upper and lower bounds on $H(A_n)$ in~\eqref{eq:secondorderupper} and~\eqref{eq:secondorderlower} for $f(n) = \Omega\left(\frac{1}{n^2}\right)$ and $f(n) = O\left(\frac{1}{n}\right)$, which implies statement~\eqref{eq:graphEntropyl2} for $f(n)$ in the same regime.

For $f(n)= \Theta(1)$, 
we have
\begin{align}
    0 \le I(X^n;A_n) &\le H(X^n) = n H(X_1)\nonumber\\ &\le n \log L=o\left(n^2f(n)\log \frac{1}{f(n)}\right),\label{eq:MI_order}
\end{align}
which shows that 
\begin{align*}
H(A_n) &= {n \choose 2}\pv^T h\bigl(f(n)\Qv\bigr)\pv + o\left(n^2 \right)\\&={n \choose 2}\pv^T h\bigl(f(n)\Qv\bigr) \pv + o\left(n^2f(n)\log\frac{1}{f(n)}\right).
\end{align*}

Next, consider the case when $f(n) = \omega\left(\frac{1}{n}\right)$ and $f(n) = o(1)$. By equation~\eqref{eq:phqp-expand}, we have
\begin{align}
    &\pv^T h\bigl(f(n)\Qv\bigr)\pv\nonumber\\&= f(n)\bigg(\log \left(\frac{1}{f(n)}\right)\PQP+\PQP\log e\nonumber\\ &+\mathbf{p}^T\mathbf{Q^*p}+o(1)\bigg)\nonumber\\
    &\stackrel{(a)}{=} f(n)\left(\log \left(\frac{1}{f(n)}\right)\PQP\right)(1+o(1)),\label{eq:phqp-first-order}
\end{align}
where (a) follows since $\log \left(\frac{1}{f(n)}\right)=\omega(1)$ and $\mathbf{p}^T\mathbf{Qp}\log\frac{1}{\PQP}=\Theta(1)$ and $\mathbf{p}^T\mathbf{Q^*p}=\Theta(1)$.

Since $f(n)=\omega\left(\frac{1}{n}\right)$, we have 
\begin{align}
    I(X^n;A_n) &\le H(X^n) = n H(X_1)\nonumber\\ &\le n \log L=o\left(f(n)\log\frac{1}{f(n)}n^2\right),\label{eq:case2MI}
\end{align}
which proves~\eqref{eq:graphEntropyl2}.
Substituting~\eqref{eq:phqp-first-order} and~\eqref{eq:case2MI} into~\eqref{eq:upperboundEntropyRateCommLabels} yields 
\begin{align*}
H(A_n)=&\binom{n}{2}f(n)\bigg(\log\left(\frac{1}{f(n)}\right)\PQP+\PQP\log e\\ &+\PsQP+o(1)\bigg).
\end{align*}

Finally, consider the case when $f(n) = \Omega\left(\frac{1}{n^2}\right)$ and $f(n) = O\left(\frac{1}{n}\right)$. By properties of the entropy, we have
\begin{align}\label{eq:EntropyBdSadwich}
  H(A_n|X^n) \le  H(A_n) \le {n \choose 2} H(A_{12}).
\end{align}
By Lemma~\ref{lem:technical}, we know that
\begin{align}\label{eq:EntropyTightBds}
  {n \choose 2} \pv^T h(f(n)\Qv) \pv \le  H(A_n) \le {n \choose 2} h\left(f(n)\pv^T\Qv\pv\right).
\end{align}
By~\eqref{eq:hpqp-expand} and~\eqref{eq:phqp-expand}, we can rewrite the upper and lower bounds as
\begin{align}
    &{n \choose 2} h\left(f(n)\pv^T\Qv\pv\right)\nonumber\\&=\binom{n}{2}f(n)\bigg(\log \left(\frac{1}{f(n)}\right)\PQP+\PQP\log e\nonumber\\ &+\mathbf{p}^T\mathbf{Qp}\log\frac{1}{\PQP}+o(1)\bigg).\label{eq:case3upper}
\end{align}
and
\begin{align}
    &\binom{n}{2}\mathbf{p^T}h(f(n)\mathbf{Q})\mathbf{p}\nonumber\\ &=\binom{n}{2}f(n)\bigg(\log \left(\frac{1}{f(n)}\right)\PQP+\PQP\log e\nonumber\\&+\mathbf{p}^T\mathbf{Q^*p}+o(1)\bigg),\label{eq:case3lower}
\end{align}
By substituting~\eqref{eq:case3lower} and~\eqref{eq:case3upper} into~\eqref{eq:EntropyTightBds}, we get bounds~\eqref{eq:secondorderupper} and~\eqref{eq:secondorderlower} in this regime. Next, we are going to verify~\eqref{eq:graphEntropyl2} in this regime. 
To see that, we take~\eqref{eq:case3upper} minus~\eqref{eq:case3lower} and get
\begin{align*}
    &{n \choose 2} h\left(f(n)\pv^T\Qv\pv\right)-{n \choose 2} \pv^T h(f(n)\Qv) \pv\\&= {n \choose 2} f(n)\left(\PQP\log\frac{1}{\PQP}-\PsQP+o(1)\right)\\
    &\stackrel{(a)}{=}o\left(n^2 f(n)\log\frac{1}{f(n)}\right),
\end{align*}
where (a) follows because $f(n)=o(1)$ and $\left(\PQP\log\frac{1}{\PQP}-\PsQP\right)=\Theta(1)$.
\end{proof}

\subsection{Asymptotic i.i.d.\! via Submatrix Decomposition}
We first invoke a known property of stochastic block models (see, for example,~\cite{Lauritzen--Rinaldo--Sadeghi2017, Orbanz--Roy2014}). We include the proof here for completeness.
\begin{lem}[Exchangeability of SBM]\label{lem:exchangeability}
 Let $A_n \sim \mathrm{SBM}(n,L,\pv,\Wv)$. For a permutation $\pi: [n] \to [n]$, let $\pi(A_n)$ be an $n \times n$ matrix whose $(i,j)$ entry is given by $A_{\pi(i),\pi(j)}$. Then, for any permutation $\pi: [n] \to [n]$, the joint distribution of $A_n$ is the same as the joint distribution of $\pi(A_n)$, i.e.,
 \begin{equation}
  A_n \stackrel{d}{=} \pi(A_n).
 \end{equation}

\end{lem}
\begin{proof}Let $a_n$ be a realization of the random matrix $A_n$ and $\pi(X^n)$ be the permuted vector $(X_{\pi(1)}, \ldots, X_{\pi(n)})$. For any symmetric binary matrix $a_n$ with zero diagonal entries, we have
{\allowdisplaybreaks
\begin{align*}
 &\P(A_n = a_n) \\&= \sum_{x^n \in [L]^n} \P(A_n = a_n,X^n = x^n)\\
 &= \sum_{x^n \in [L]^n} \P(A_n = a_n | X^n = x^n) \prod_{i=1}^n \P(X_i = x_i)\\
 &\stackrel{(a)}{=} \sum_{x^n \in [L]^n} \prod_{\substack{i,j\\ 1 \le i < j \le n}} \P(A_{ij} = a_{ij}|X_i = x_i, X_j = x_j) \\&\hspace{5em}\cdot\prod_{i=1}^n \P(X_{\pi(i)} = x_i)\\
 &\stackrel{(b)}{=} \sum_{x^n \in [L]^n}\prod_{\substack{i,j\\ 1 \le i < j \le n}} (W_{x_i,x_j})^{a_{ij}}(1-W_{x_i,x_j})^{1-a_{ij}}
 \\&\hspace{5em}\cdot\prod_{i=1}^n \P(X_{\pi(i)} = x_i)\\
 &\stackrel{(c)}{=} \sum_{x^n \in [L]^n} \prod_{\substack{i,j\\ 1 \le i < j \le n}} \P(A_{\pi(i),\pi(j)} = a_{ij}|X_{\pi(i)} = x_i, X_{\pi(j)} = x_j) 
 \\&\hspace{5em}\cdot\prod_{i=1}^n \P(X_{\pi(i)} = x_i) \\
 &{=} \sum_{x^n \in [L]^n} \P(\pi(A_n) = a_n, \pi(X^n) = x^n)\\
 &= \P(\pi(A_n) = a_n),
\end{align*}}%
where $(a)$ follows since $X^n$ are i.i.d. and thus $\P(X_i = x_i) = \P(X_{\pi(i)} = x_i)$ and $(b)$ follows since $A_{ij} \sim \Bern(W_{X_i,X_j})$, and thus 
\begin{align}
 &\P(A_{ij} = a_{ij}|X_i = x_i, X_j = x_j)\nonumber\\ &= \begin{cases}
                                             W_{x_i,x_j} & \text{ if } a_{ij} = 1\\
                                             1-W_{x_i,x_j} & \text{ if } a_{ij} = 0
                                            \end{cases}\label{eq:condProb1}\\
           &= (W_{x_i,x_j})^{a_{ij}}(1-W_{x_i,x_j})^{1-a_{ij}}. \label{eq:condProb2}
\end{align}
The step in $(c)$ follows since $A_{\pi(i),\pi(j)} \sim \Bern(W_{X_{\pi(i)},X_{\pi(j)}})$ and the conditional probability has the same expression as in~\eqref{eq:condProb2}.
\end{proof}

Now we are ready to establish Proposition~\ref{prop:blkIndep}, which is restated as follows.
\decomp*
\begin{proof}[\bf Proof of Proposition~\ref{prop:blkIndep}]
For any $i_1 \neq j_1$ and $i_2 \neq j_2$, $1 \le i_1,j_1,i_2,j_2 \le n'$,  consider a permutation $\pi_1: [n] \to [n]$ that has 
\[
   \pi_1(x) =\Big\{ \begin{array}{lr}
       x + (i_2-i_1)k & \text{for } (i_1-1)k + 1 \leq x \leq i_1k\\
       x + (j_2-j_1)k & \text{for } (j_1-1)k + 1 \leq x \leq j_1k 
        \end{array}
  \]
  and the remaining $n-2k$ arguments are mapped to the $n-2k$ values in $[n]\setminus \{(i_2-1)k + 1 ,\ldots, i_2k, (j_2-1)k ,\ldots ,j_2k\}$ in any order. Lemma~\ref{lem:exchangeability} implies that $\BB_{i_1,j_1}$, which is the submatrix formed by the rows 
  $(i_1-1)k + 1, \ldots, i_1k$ and the columns $(j_1-1)k + 1,\ldots, j_1k$ has the same distribution  as the submatrix formed by the rows $\pi_1((i_1-1)k + 1),\ldots, \pi_1( i_1k)$ and the columns $\pi_1((j_1-1)k + 1),\ldots, \pi_1(j_1k)$. From the definition of $\pi_1$, we see that the latter submatrix is $\BB_{i_2,j_2}$ and we establish that $\BB_{i_1,j_1} \stackrel{d}{=} \BB_{i_2,j_2}$. Similarly, defining a permutation $\pi_2: [n] \to [n]$ which has 
  \[
   \pi_2(x) = \begin{array}{lr}
       x + (l_2-l_1)k & \text{for } (l_1-1)k + 1 \leq x \leq l_1k
        \end{array}
  \]
  and invoking Lemma~\ref{lem:exchangeability} establishes $ \BB_{l_1,l_1} \stackrel{d}{=} \BB_{l_2,l_2}$.
 
 Now, clearly $H(\BB_{\mathrm{ut}}) \le {n' \choose 2}H(\BB_{12})$, and therefore we have 
\begin{equation}\label{eq:limsupBut}
    \limsup_{n \to \infty} \frac{H(\BB_{\mathrm{ut}})}{{n' \choose 2}H(\BB_{12})} \le 1.
\end{equation}
Moreover, we have $H(A_n) = H(\BB_{\mathrm{ut}}, \BB_{\mathrm{d}}) \le H(\BB_{\mathrm{ut}}) + H(\BB_{\mathrm{d}}) \le H(\BB_{\mathrm{ut}}) + n'H(\BB_{\mathrm{11}}) \le  H(\BB_{\mathrm{ut}}) + n'k^2 H(A_{12})$ where the last inequality follows by noting that except for the diagonal elements of $\BB_{\mathrm{d}}$ (which are zero and thus have zero entropy), all other elements have the same distribution as $A_{12}$. We therefore obtain
$H(\BB_{\mathrm{ut}}) \ge H(A_n) - n'k^2H(A_{12}) \ge H(A_n) - nk H(A_{12}) \ge H(A_n|X_1^n) - nk H(A_{12}) = {n \choose 2}\pv^T h(f(n)\Qv)\pv - nk h(f(n)\pv^T \Qv\pv)$, where the last equality follows by Lemma~\ref{lem:technical}. Consequently,
\begin{align}\label{eq:lbBut}
   \frac{H(\BB_{\mathrm{ut}})}{{n' \choose 2}H(\BB_{12})} \ge \frac{{n \choose 2} \left(\pv^T h(f(n)\Qv) \pv - \frac{2 k h(f(n)\pv^T \Qv\pv) }{n-1}\right)}{{n' \choose 2}H(\BB_{12})}.
\end{align} 
We will now analyze the right-hand side of~\eqref{eq:lbBut} in two parameter regimes. 
{\allowdisplaybreaks
\begin{itemize}
    \item \underline{$f(n) = \Theta(1)$ :} We have 
\begin{align}
    H(\BB_{12}) &\stackrel{(a)}{\le} H(\BB_{12}|X_1^{2k}) + H(X_1^{2k}) \nonumber \\
    &= H(\BB_{12}|X_1^{2k}) + 2k H(\pv) \nonumber\\
    &\stackrel{(b)}{=} k^2 H(A_{1,k+1}|X_1,X_{k+1}) + 2k H(\pv) \nonumber\\
    &\stackrel{(c)}{\le} k^2\left(\pv^T h(\Qv)\pv + 2 \frac{\log L}{k}\right), \label{eq:BlockEntropyUB}
\end{align}
where (a) follows from the chain rule, (b) follows since all elements of the matrix $\BB_{12}$ are independent given $X_1,\cdots,X_{2k}$ and (c) follows by Lemma~\ref{lem:technical}.
    Plugging this into the right-hand side of~\eqref{eq:lbBut} we obtain 
    \begin{align}
        \frac{H(\BB_{\mathrm{ut}})}{{n' \choose 2}H(\BB_{12})} 
         \ge\frac{{n \choose 2}\left(\pv^T h(\Qv)\pv - \frac{2k h(\pv^T \Qv\pv)}{n-1}\right)}{{n' \choose 2}k^2\left(\pv^T h(\Qv)\pv + 2 \frac{\log L}{k}\right)}. \label{eq:liminfBut}
    \end{align}
    Since $k = o(n), k = \omega(1)$ and ${n' \choose 2}k^2 \sim {n \choose 2}$, we have from~\eqref{eq:liminfBut} 
    \begin{align}\label{eq:liminfone}
          \liminf_{n \to \infty} \frac{H(\BB_{\mathrm{ut}})}{{n' \choose 2}H(\BB_{12})} \ge 1,
    \end{align}
    which together with~\eqref{eq:limsupBut} yields the required result.

    \item \underline{$f(n) = \Omega\left(\frac{1}{n^2}\right), f(n) = o(1):$} Since $\BB_{12}$ is a matrix of $k^2$ identically distributed Bernoulli random variables, we have 
    \begin{align}\label{eq:vanishingPblkEntropy}
    H(\BB_{12}) \le k^2 h(A_{1,k+1}) = k^2 h\left(f(n)\pv^T \Qv \pv \right).
    \end{align}
    Plugging this into the RHS of~\eqref{eq:lbBut} then yields 
    \begin{align}
        &\frac{H(\BB_{\mathrm{ut}})}{{n' \choose 2}H(\BB_{12})} \ge \frac{{n \choose 2}\left(\pv^T h(f(n)\Qv)\pv - \frac{2k h(f(n)\pv^T \Qv\pv)}{n-1}\right)}{{n' \choose 2}k^2 h\left(f(n)\pv^T \Qv\pv\right)}. \label{eq:bbutVanishP}
    \end{align}
    We first observe that in this parameter range, since $f(n) = o(1)$, 
    by Lemma~\ref{lem:technical}, we have 
    \begin{align}\label{eq:entropyGapVanishes}
     \pv^T h(f(n)\Qv)\pv   =(1+o(1))  h\left(f(n)\pv^T \Qv\pv\right).
    \end{align}
     Finally using that $k = o(n)$ and ${n' \choose 2}k^2 \sim {n \choose 2}$, we establish
     \begin{align}\label{eq:liminfoneVanishing}
          \liminf_{n \to \infty} \frac{H(\BB_{\mathrm{ut}})}{{n' \choose 2}H(\BB_{12})} \ge 1,
    \end{align}
    which together with~\eqref{eq:limsupBut} yields the required result. 
\end{itemize}}%
\end{proof}

\subsection{Length of the Laplace Probability Assignment}
In this section, we prove Proposition~\ref{prop:NonIIDcomp}, which is restated as follows.
\laplace*
\begin{proof}[\bf Proof of Proposition~\ref{prop:NonIIDcomp}]
The joint probability of Laplace probability assignment given in equation~\eqref{eq:LaplaceProbAssgn} is a classic result. Nonetheless, we provide the derivation of~\eqref{eq:LaplaceProbAssgn} in Appendix~\ref{sec:joint-dist-lap-kt}.

Now, we move on to prove the expected length upper bound ~\eqref{eq:LenLaplaceGen}. Let us first elaborate on the relation between probability assignment and compression length. 
In Algorithm~\ref{alg:arithmetic}, the terms $\log(q(x_{j+1}|x^j))$ are added up, which leads to the marginal probability implied by the sequential probability assignment
\begin{equation}\label{eq:marginal}
 \sum_{j=0}^{N-1}\log(q(x_{j+1}|x^j)) = \log \left(\prod_{j=0}^{N-1} q(x_{j+1}|x^j)\right) = \log(q(x^N)).
\end{equation}
The compression output length of Algorithm~\ref{alg:arithmetic} is $\left\lceil\log\frac{1}{q(x^N)}\right\rceil + 1$. 

Now we analyze the compression length of the Laplace compressor for the sequence $Z_1,Z_2,\ldots,Z_N$. 
Define $\theta_i \defeq \P(Z_1 = i), N_i \defeq \sum_{k=1}^N \indi \{Z_k = i\}, i \in [m]$. 
We have 
{\allowdisplaybreaks
\begin{align}
      &\log \frac{1}{q_{\Lap}(z^N)}\nonumber\\&= \log \frac{\theta_1^{N_1}\theta_2^{N_2}\cdots\theta_m^{N_m}}{q_{\Lap}(z^N)} + \log \frac{1}{\theta_1^{N_1}\theta_2^{N_2}\cdots\theta_m^{N_m}} \nonumber \\
      &= \log {N+m-1 \choose m-1} + \log \left(\tfrac{N!}{N_1!N_2!\cdots N_m!}\theta_1^{N_1}\theta_2^{N_2}\cdots\theta_m^{N_m}\right) \nonumber\\ &\;\;\;+  \log \frac{1}{\theta_1^{N_1}\theta_2^{N_2}\cdots\theta_m^{N_m}} \nonumber \\
      &\stackrel{(a)}{\le} \log {N+m-1 \choose m-1}  + \log \frac{1}{\theta_1^{N_1}\theta_2^{N_2}\cdots\theta_m^{N_m}} \nonumber  \\
      &\stackrel{(b)}{\le} (m-1)\log\left(e\left(\frac{N}{m-1}+1\right)\right) + \log \frac{1}{\theta_1^{N_1}\theta_2^{N_2}\cdots\theta_m^{N_m}} \nonumber \\
      &\le m \log (2eN) + \sum_{i=1}^m N_i \log\frac{1}{\theta_i}, \label{eq:LenLaplaceRV} 
\end{align}}%
where (a) follows since $\frac{N!}{N_1!N_2!\cdots N_m!}\theta_1^{N_1}\theta_2^{N_2}\cdots\theta_m^{N_m}$ is a multinomial probability which is always upper bounded by 1, and (b) follows since ${n \choose k} \le \left(\frac{en}{k}\right)^k$. Taking expectation on both sides of~\eqref{eq:LenLaplaceRV}, we obtain 
{\allowdisplaybreaks
\begin{align}
    &\E\left[ \log \frac{1}{q_{\Lap}(Z^N)} \right]\nonumber \\ &\le  m \log (2e N) + \sum_{i=1}^m \E[N_i] \log\frac{1}{\theta_i} \nonumber \\
    &= m \log (2e N) + \sum_{i=1}^m \E\left[\sum_{k=1}^N\indi\{Z_k = i\}\right] \log\frac{1}{\theta_i}\\
    &\stackrel{(c)}{=} m \log (2e N) + \sum_{i=1}^m \left(\sum_{k=1}^N\E\left[\indi\{Z_k = i\}\right]\right)\log\frac{1}{\theta_i}\\
    &\stackrel{(d)}{=} m \log (2e N) + \sum_{i=1}^m N \P(Z_1 = i)\log\frac{1}{\theta_i}\\
    &=  m \log (2eN) + \sum_{i=1}^m N \theta_i \log\frac{1}{\theta_i} \nonumber \\
    &=  m \log (2eN) + N H(Z_1), \nonumber
\end{align}}%
where (c) follows by the linearity of expectation and (d) follows because all $Z_i$ are identically distributed.
Finally, we have 
$$\E[\ell_{\Lap}(Z^N)]\le \E\left[ \log \frac{1}{q_{\Lap}(Z^N)} \right]+2 \le m\log(2eN) + N H(Z_1)+2$$
as required.
\end{proof}

\subsection{Length of the KT probability assignment}
To prove Proposition~\ref{prop:KTlength}, we first state a technical lemma.
\begin{lem}\label{lem:KTlength}
 For any integer $m>0$, $N_1, N_2, \cdots N_m \in \mathbb{N}$ and probability distribution $(\theta_1, \cdots \theta_m)$, $$\frac{\binom{N}{N_1,N_2\cdots N_m}\theta_1^{N_1}\cdots\theta_m^{N_m}}{\binom{2N}{2N_1,2N_2\cdots 2N_m}\theta_1^{2N_1}\cdots\theta_m^{2N_m}}\ge 1,$$ where $N=\sum_{i=1}^m N_i$.
\end{lem}
We defer the proof of Lemma~\ref{lem:KTlength} to Appendix~\ref{sec:lem-kt-proof}.
\begin{rem}
 Equivalently, consider an urn containing a known number of balls with $m$ different colours. The lemma claims that the probability of getting $N_1$ balls of colour 1, $N_2$ of balls of colour 2, $\cdots$ $N_m$ balls of colour m out of $N$ draws with replacement is always greater than the probability of getting $2N_1$ balls of colour 1, $2N_2$ of balls of colour 2, $\cdots$ $2N_m$ balls of colour m out of $2N$ draws with replacement.
\end{rem}
Now, we are ready to Proposition~\ref{prop:KTlength}, which is restated as follows.
\kt*

\begin{proof}[\bf Proof of Proposition~\ref{prop:KTlength}]
The joint probability of KT probability assignment given in equation~\eqref{eq:KTProbAssgn} is a classic result. Nonetheless, we provide the derivation of~\eqref{eq:KTProbAssgn} in Appendix~\ref{sec:joint-dist-lap-kt}.

Now, we move on to prove the expected length upper bound~\eqref{eq:LenKT}. In this proof, we define a generalized form of the factorial function. Let $x$ be a positive integer, $(x+\frac{1}{2})!=\frac{1}{2} \frac{3}{2} \cdots (x+\frac{1}{2})$. Since $(2N_1-1)!!=\frac{(2N_1)!}{2^{N_1}(N_1)!}$, we have
{\allowdisplaybreaks
\begin{align*} 
&m(m+2)\cdots (m+2N-2)\\&=2^N \left(\frac{m}{2}\right) \left(\frac{m+2}{2}\right)\cdots \left(\frac{m+2N-2}{2}\right)\\&=2^N\frac{\left(\frac{m}{2}+N-1\right)!}{\left(\frac{m}{2}-1\right)!}.
\end{align*}}%
Therefore we can rewrite the KT probability assignment in~\eqref{eq:KTProbAssgn} as
{\allowdisplaybreaks
\begin{align*} 
&q_\mathrm{KT}(z^N)\\&=\frac{(\frac{m}{2}-1)!}{2^N(\frac{m}{2}+N-1)!}\frac{\binom{2N}{N}}{\binom{2N}{N}}\prod_{i=1}^m \frac{(2N_i)!}{N_i ! \vspace{.2em}2^{N_i}}\\
&=\frac{(\frac{m}{2}-1)!}{2^N(\frac{m}{2}+N-1)!} \binom{2N}{N}N! \frac{N!}{(2N)!}\prod_{i=1}^m \frac{(2N_i)!}{N_i !2^{N_i}} \\
&\stackrel{(a)}{\ge}\frac{\left(\frac{m}{2}-1\right)!\binom{2N}{N}}{4^N (N+\frac{m}{2}-\frac{1}{2})^{\frac{m-1}{2}}}\frac{N!}{(2N)!}\prod_{i=1}^m \frac{(2N_i)!}{N_i !}\\
&\stackrel{(b)}{=}\frac{\theta_1^{N_1}\cdots\theta_m^{N_m}(\frac{m}{2}-1)!\binom{2N}{N}}{4^N (N+\frac{m}{2}-\frac{1}{2})^{\frac{m-1}{2}}}\frac{\binom{N}{N_1,N_2\cdots N_m}\theta_1^{N_1}\cdots\theta_m^{N_m}}{\binom{2N}{2N_1,2N_2\cdots 2N_m}\theta_1^{2N_1}\cdots\theta_m^{2N_m}},
\end{align*}}%
where (a) follows that when m is even, $\frac{N!}{(\frac{m}{2}+N-1)!}=\frac{1}{(N+1)\cdots (\frac{m}{2}+N-1)}\ge \frac{1}{(N+\frac{m}{2}-\frac{1}{2})^{\frac{m-1}{2}}}$ and when m is odd, $\frac{N!}{(\frac{m}{2}+N-1)!}\ge\frac{N!}{(\frac{m}{2}+N-\frac{1}{2})!}=\frac{1}{(N+1)\cdots (\frac{m}{2}+N-\frac{1}{2})}\ge\frac{1}{(N+\frac{m}{2}-\frac{1}{2})^{\frac{m-1}{2}}}$, (b) follows that $\binom{N}{N_1,N_2\cdots N_m}=\frac{N!}{\prod_{i=1}^m N_i !}$ and $\theta_i\triangleq\mathbb{P}(Z_1=i)$. By lemma~\ref{lem:KTlength}, we have $q_\mathrm{KT}(z^N)\ge\frac{\theta_1^{N_1}\cdots\theta_m^{N_m}(\frac{m}{2}-1)!\binom{2N}{N}}{4^N (N+\frac{m}{2}-\frac{1}{2})^{\frac{m-1}{2}}}$. Thus, 
{\allowdisplaybreaks
\begin{align*} 
&\log \frac{1}{q_\mathrm{KT}(z^N)}\\&\le \log{\frac{1}{\theta_1^{N_1}\cdots\theta_m^{N_m}}}+\log{\frac{4^N (N+\frac{m}{2}-\frac{1}{2})^{\frac{m-1}{2}}}{(\frac{m}{2}-1)!\binom{2N}{N}}}\\
&=\log{\frac{1}{\theta_1^{N_1}\cdots\theta_m^{N_m}}}+\frac{m-1}{2}\log{\left(N+\frac{m-1}{2}\right)}\\&\;\;\;+\log{\frac{4^N}{\binom{2N}{N}}}-\log{\left(\frac{m}{2}-1\right)!}\\
&\stackrel{(a)}{\le}\log{\frac{1}{\theta_1^{N_1}\cdots\theta_m^{N_m}}}+\frac{m-1}{2}\log{\left(N+\frac{m-1}{2}\right)}\\&\;\;\;+\log{\frac{4^N}{\binom{2N}{N}}}-\left(\frac{m}{2}-1\right)\log{(\frac{\frac{m}{2}-1}{e})}\\
&\stackrel{(b)}{\sim}\log{\frac{1}{\theta_1^{N_1}\cdots\theta_m^{N_m}}}+\frac{m-1}{2}\log{\left(N+\frac{m-1}{2}\right)}\\&\;\;\;+\log{\sqrt{\pi N}}-\left(\frac{m}{2}-1\right)\log{\left(\frac{\frac{m}{2}-1}{e}\right)}\\
&\sim \frac{m}{2}\log{\frac{e(\frac{m}{2}+N)}{m/2}}+\log{\sqrt{\pi N}}+\log{\frac{1}{\theta_1^{N_1}\cdots\theta_m^{N_m}}}\\
&=\frac{m}{2}\log\left(e\bigl(1+\tfrac{2N}{m}\bigr)\right)+\frac 12\log(\pi N)+\sum_{i=1}^m N_i \log{\frac{1}{\theta_i}},
\end{align*}}%
where $(a)$ follows Stirling's approximation $k!\ge \sqrt{2\pi k}(\frac{k}{e})^{k}e^{\frac{1}{12k+1}}$ and (b) follows Stirling's approximation for binomial coefficient, i.e., $\binom{2N}{N}\sim \frac{4^N}{\sqrt{\pi N}}$. Therefore, we have
\begin{align*}
&\E\left[\log \frac{1}{q_\mathrm{KT}(z^N)}\right]\\&\le \frac{1}{2}m\log\left(e\bigl(1+\tfrac{2N}{m}\bigr)\right)+\frac 12\log(\pi N)+NH(Z_1).
\end{align*}
Finally, it follows that
\begin{align*}
&\E[\ell_{\mathrm{KT}}(Z^N)]\le\E\left[\log \frac{1}{q_\mathrm{KT}(z^N)}\right]+2\\&\le \tfrac{m}{2}\log\left(e\bigl(1+\tfrac{2N}{m}\bigr)\right)+\tfrac 12\log(\pi N)+NH(Z_1)+2.
\end{align*}
\end{proof}

\section{Minimax Redundancy Analysis}\label{sec:proveminimax}
So far, we have proved the universality of our proposed algorithms. For this, we showed that the expected length of the proposed compressor matches the first-order term in the graph entropy. In this section, we study the redundancy of the proposed compressor $C_k$ and prove the minimax redundancy result Theorem~\ref{thm:minimaxredun}:
\Minimax*
Towards that goal, we first state a proposition that upper bounds the redundancy of compressor $C_k$.

\begin{prop} 
\label{prop:redundancyStrong}
For every $0 < \eps < 1$, let $A_n \sim \mathrm{SBM}(n,L,\mathbf{p},f(n)\mathbf{Q}) \in \mathscr{P}_1(\eps)$. 
Let $k=\omega(1)$ and $k\le \sqrt{\delta\log n}$ for some $0<\delta<\epsilon$.
\begin{itemize}
    \item
If $f(n) = o(1)$ and $f(n) = \Omega\left(\frac{1}{n^{2-\eps}}\right)$,
then the expected length of compressor $C_k$ defined in Section~\ref{sec:univGraphComp} is upper bounded as
\begin{align}
    &\E[\ell(C_k(A_n))]-H(A_n)\nonumber\\&\le 
     \binom{n}{2}f(n)\left(\PQP\log\left( \frac{1}{\PQP}\right)-\mathbf{p}^T\mathbf{Q^*p}\right)\nonumber\\&\;\;\;+o(n^2f(n))\label{eq:redundancy-vanishing}\\
    &= o(H(A_n)),\nonumber
\end{align}
where $\mathbf{Q^*}$ denotes an $L\times L$ matrix whose $(i,j)$ entry is $Q_{ij}\log (\frac{1}{Q_{ij}})$ when $Q_{ij}\neq 0$ and $0$ when $Q_{ij}=0$.

\item If $f(n) = \Theta(1)$, then the expected length of compressor $C_k$ defined in Section~\ref{sec:univGraphComp} is upper bounded as
\begin{align}
\label{eq:redundancy-constant}
\E(\ell(C_k(A_n)))-H(A_n) &\le H(\mathbf{p})\frac{n^2}{k}+o\left(\frac{n^2}{k}\right)\\
&= o(H(A_n)),\nonumber
\end{align}
where $H(\mathbf{p}) = \sum_{i=1}^Lp_i \log \left(\frac{1}{p_i}\right)$.
\end{itemize}
\end{prop}

Because $C_k$ is a valid compressor for $\mathscr{P}_3$, we can get an upper bound for the minimax redundancy by maximizing the right-hand side of equation~\eqref{eq:redundancy-vanishing}. For the lower bound, we will apply the \emph{redundancy-capacity theorem} in~\cite{Merhav--Feder1995}.
\begin{proof}[\bf Proof of Theorem~\ref{thm:minimaxredun}]
Firstly, we prove the lower bound for the minimax redundancy. The proof essentially follows from the redundancy-capacity theorem~\cite{Merhav--Feder1995}. To start with, we will define a vector $\Thetav$ of parameters $L,\pv$ and $\Qv$ for SBM. Consider stochastic block model SBM$(n,L,\pv,f(n)\Qv)$. Recall that $L$ is a scalar. Vector $\pv=(p_1,\ldots,p_L)^T$ represents a distribution over $[L]$, so it can be written as $\pv=(p_1,\ldots,p_{L-1},1-\sum_{i=1}^{L-1}p_i)^T$. For the $L$-by-$L$ symmetric matrix $\Qv$, we can write the upper-triangle entries into a vector $(Q_{11},\ldots,Q_{1L},Q_{22},\ldots,Q_{2L},\ldots,Q_{LL})$. We define the parameter vector
\begin{align*}
&\Thetav(\mathrm{SBM}(n,L,\pv,f(n)\Qv))\\&=(L,p_1,\ldots,p_{L-1},Q_{11},\ldots,Q_{1L},Q_{22},\ldots,Q_{2L},\ldots,Q_{LL}).
\end{align*}
Notice that by our definition, $\Thetav$ is a variable-length vector. Its length $L+\frac{L(L+1)}{2}$ is a function of its first entry $L$. For some fixed $f(n)$ and $Q_\mathrm{max}$, let 
\begin{align*}
\mathcal{W}:=&\{\Thetav(\mathrm{SBM}(n,L,\pv,f(n)\Qv)):\\&\mathrm{SBM}(n,L,\pv,f(n)\Qv)\in\mathscr{P}_3(f(n),Q_\mathrm{max})\}
\end{align*}
denote the set of all parameter vectors for the family of SBMs in $\mathscr{P}_3(f(n),Q_\mathrm{max})$. Let $\mu(\cdot)$ denote a probability distribution supported on $\mathcal{W}$. Suppose $A_n\sim \mathrm{SBM}(n,L,\pv,f(n)\Qv)$ with random SBM parameters following prior $\mu$, by redundancy-capacity theorem (see, for example,~\cite{Merhav--Feder1995}), we have
$$R^*_n(\mathscr{P}_3(f(n),Q_\mathrm{max}))=\sup_\mu I_\mu(\Thetav;A_n),$$
where $I_\mu(\Thetav;A_n)$ denote the mutual information between $\Thetav$ and $A_n$ under the prior $\mu$ and the supremum is taken over all possible probability distributions supported on $\mathcal{W}$. To lower bound the quantity $\sup_\mu I_\mu(\Thetav;A_n)$, we consider a particular distribution $\nu\in\mathcal{W}$ such that $L$ is fixed to $1$ and $Q_{11}\sim\mathrm{Unif}(0,Q_\mathrm{max})$, i.e, the \erdos--\renyi model with edge probability being uniform from $(0,f(n)Q_\mathrm{max}]$. Since $\nu$ is supported on $\mathcal{W}$, we have $\sup_\mu I_\mu(\Thetav;A_n)\ge I_\nu(\Thetav;A_n)$. 

We now move on to evaluate the mutual information $I_\nu(\Thetav;A_n)$. For an adjacency matrix $A_n$, we define $\hat{Q}_{11}:=\frac{K}{\binom{n}{2}f(n)}$ where $K=\sum_{i=1}^{n-1}\sum_{j=i+1}^n A_{ij}$ denotes the number of ones in upper-triangle of $A_n$. We have
\begin{align*}
    I_\nu(\Thetav;A_n)&=I_\nu(Q_{11};A_n)\\
    &=h(Q_{11})-h(Q_{11}|A_n)\\
    &\stackrel{(a)}{\ge}\log Q_\mathrm{max}-h(Q_{11}|\hat{Q}_{11})\\
    &=\log Q_\mathrm{max}-h(Q_{11}-\hat{Q}_{11}|\hat{Q}_{11})\\
    &\ge \log Q_\mathrm{max}-h(Q_{11}-\hat{Q}_{11})\\
    &\stackrel{(b)}{\ge}\log Q_\mathrm{max}-\frac12\log\left(2\pi e\mathrm{Var}(Q_{11}-\hat{Q}_{11})\right),
\end{align*}
where (a) follows since $\hat{Q}_{11}$ is a function of $A_n$ and (b) follows since Gaussian distribution maximizes differential entropy for a given variance. To evaluate the variance $\mathrm{Var}(Q_{11}-\hat{Q}_{11})$, we consider the conditional variance 
{\allowdisplaybreaks
\begin{align*}
    &\mathrm{Var}(Q_{11}-\hat{Q}_{11}|Q_{11}=q)\\&=\E\left[\left(\frac{K}{\binom{n}{2}f(n)}-Q_{11}\right)^2\Biggr|Q_{11}=q\right]\\
    &=\frac{1}{\binom{n}{2}^2(f(n))^2}\E\left[\left(K-\binom{n}{2}f(n)q\right)^2\Biggr|Q_{11}=q\right]\\
    &\stackrel{(c)}{=}\frac{\binom{n}{2}f(n)q(1-f(n)q)}{\binom{n}{2}^2(f(n))^2}\\
    &=\frac{q(1-f(n)q)}{\binom{n}{2}f(n)},
\end{align*}}%
where (c) follows since $K|Q_{11}=q\sim \mathrm{Binom}(\binom{n}{2},f(n)q)$. By law of total variance and since $\E[\hat{Q}_{11}-Q_{11}|Q_{11}=q]=0$, we have
\begin{align*}
 \mathrm{Var}(Q_{11}-\hat{Q}_{11})&=\E\left[\frac{Q_{11}(1-f(n)Q_{11})}{\binom{n}{2}f(n)}\right]\\&\le \frac{\E[Q_{11}]}{\binom{n}{2}f(n)}=\frac{Q_\mathrm{max}}{2\binom{n}{2}f(n)}.
 \end{align*}
Finally, we have 
\begin{align*}
    I_\nu(\Thetav;A_n)&\ge\log Q_\mathrm{max}-\frac12\log\left(2\pi e\mathrm{Var}(Q_{11}-\hat{Q}_{11})\right)\\
    &\ge \log Q_\mathrm{max}-\frac12\log\left(\frac{Q_\mathrm{max}\pi e}{\binom{n}{2}f(n)}\right)\\
    &=\frac12\log\left(\frac{\binom{n}{2}f(n)Q_\mathrm{max}}{\pi e}\right),
\end{align*}
which proves the lower bound.

Now, we prove the upper bound for the minimax redundancy. Let $\mathcal{S}(f(n),Q_\mathrm{max})$ denote the set of all triples $(L,\mathbf{p},\mathbf{Q})$ such that $\mathrm{SBM}(n,L,\pv,f(n)\Qv)\in \mathscr{P}_3(f(n),Q_\mathrm{max})$. Recall that the minimax redundancy is defined as $R^*_n(\mathscr{P}_3(f(n),Q_\mathrm{max}))=\inf\limits_C\sup\limits_{(L,\mathbf{p},\mathbf{Q})\in \mathcal{S}(f(n),Q_\mathrm{max})}(\E[\ell(C(A_n))]-H(A_n))$. Since our proposed compressor $C_k$ is a valid lossless compressor for $\mathscr{P}_3$, we can upper bound the minimax redundancy with the maximum redundancy of $C_k$. By Proposition~\ref{prop:redundancyStrong}, we have
\begin{align*}
    &R^*_n(\mathscr{P}_3(f(n),Q_\mathrm{max}))\\&=\inf\limits_C\sup\limits_{(L,\mathbf{p},\mathbf{Q})\in \mathcal{S}(f(n),Q_\mathrm{max})}(\E[\ell(C(A_n))]-H(A_n))\\
    &\le \sup\limits_{(L,\mathbf{p},\mathbf{Q})\in \mathcal{S}(f(n),Q_\mathrm{max})}(\E[\ell(C_k(A_n))]-H(A_n))\\
    &\stackrel{(d)}{\le} \sup\limits_{(L,\mathbf{p},\mathbf{Q})\in \mathcal{S}(f(n),Q_\mathrm{max})}\binom{n}{2}f(n)\bigg(\PQP\log \frac{1}{\PQP}\\&\hspace{9em}-\mathbf{p}^T\mathbf{Q^*p}\bigg)+o(n^2f(n)),
\end{align*}
where (d) follows since $\mathscr{P}_3$ is a sub-family of $\mathscr{P}_2$, so $f(n)=o(1)$.
Therefore, it suffices to maximize the constant term $\left(\PQP\log \frac{1}{\PQP}-\mathbf{p}^T\mathbf{Q^*p}\right)$ over all valid choices of $L,\mathbf{p},\mathbf{Q}$. We can rewrite the term as
{\allowdisplaybreaks
\begin{align*}
    &\PQP\log \frac{1}{\PQP}-\mathbf{p}^T\mathbf{Q^*p}\\
    &=\sum_{i,j\in[L]}p_ip_jQ_{ij}\log\frac{1}{\sum_{k,l\in[L]}p_kp_lQ_{kl}}\\&\;\;\;-\sum_{i,j\in [L]}p_ip_jQ_{ij}\log\frac{1}{Q_{ij}}\\
    &=Q_\mathrm{max}\Bigg(\sum_{i,j\in[L]}p_ip_j\frac{Q_{ij}}{Q_\mathrm{max}}\log\frac{1}{\sum_{k,l\in[L]}p_kp_lQ_{kl}}\\&\;\;\;-\sum_{i,j\in [L]}p_ip_j\frac{Q_{ij}}{Q_\mathrm{max}}\log\frac{1}{Q_{ij}}\Bigg)\\
    &=Q_\mathrm{max}\Bigg(\sum_{i,j\in[L]}p_ip_j\frac{Q_{ij}}{Q_\mathrm{max}}\log\left(\frac{1}{\sum_{k,l\in[L]}p_kp_l\frac{Q_{kl}}{Q_\mathrm{max}}}\frac{1}{Q_\mathrm{max}}\right)\\&\;\;\;-\sum_{i,j\in [L]}p_ip_j\frac{Q_{ij}}{Q_\mathrm{max}}\log\left(\frac{Q_\mathrm{max}}{Q_{ij}}\frac{1}{Q_\mathrm{max}}\right)\Bigg)\\
    &=Q_\mathrm{max}\Bigg(\sum_{i,j\in[L]}p_ip_j\frac{Q_{ij}}{Q_\mathrm{max}}\log\frac{1}{\sum_{k,l\in[L]}p_kp_l\frac{Q_{kl}}{Q_\mathrm{max}}}\\&\;\;\;-\sum_{i,j\in [L]}p_ip_j\frac{Q_{ij}}{Q_\mathrm{max}}\log\frac{Q_\mathrm{max}}{Q_{ij}}\Bigg)\\
    &=Q_\mathrm{max}\left(\mathbf{p}^T\mathbf{Q'p}\log\frac{1}{\mathbf{p}^T\mathbf{Q'p}}-\mathbf{p}^T\mathbf{(Q')^*p}\right),
\end{align*}}%
where $\mathbf{Q}'$ denotes the matrix $\mathbf{Q}/Q_\mathrm{max}$ and $(\mathbf{Q}')^*$ denotes an $L\times L$ matrix whose $(i,j)$ entry is $Q'_{ij}\log(\frac{1}{Q'_{ij}})$. By our assumptions, the entries in $\mathbf{Q}'$ are all non-negative and the largest entry is $1$. Since $\mathbf{p}$ is a probability distribution, we have $0<\mathbf{p}^T\mathbf{Q'p}\le 1$. Now, let us consider the function $g(x)=x\log \frac{1}{x}$. When $0\le x\le 1$, its supremum is obtained at $x=1/e$ and its infimum is obtained at $x=0$ or $x=1$. With these properties, we can bound our constant term
$$Q_\mathrm{max}\left(\mathbf{p}^T\mathbf{Q'p}\log\frac{1}{\mathbf{p}^T\mathbf{Q'p}}-\mathbf{p}^T\mathbf{(Q')^*p}\right)\le Q_\mathrm{max}\frac{1}{e}\log e.$$
In particular, the equality can be achieved when all the entries in $\mathbf{Q}'$ are either zero or one and $\mathbf{p}$ is carefully chosen such that $\mathbf{p}^T\mathbf{Q'p}=1/e$. This completes the proof for the upper bound.
\end{proof}
\begin{proof}[\bf Proof of Proposition~\ref{prop:redundancyStrong}]

We discuss the redundancy in two cases: (1) $f(n) = o(1)$ and $f(n) = \Omega(1/n^{2-\eps})$ and (2) $f(n) = \Theta(1)$. 

First assume $f(n) = o(1)$ and $f(n) = \Omega(1/n^{2-\eps})$. By Lemma~\ref{lem:graphEntropy}, we can lower bound the graph entropy
\begin{align}
    H(A_n)\ge&\binom{n}{2}f(n)\bigg(\log\left(\frac{1}{f(n)}\right)\PQP+\PQP\log e\nonumber\\ &+\PsQP+o(1) \bigg).\label{eq:entropylowerbound}
\end{align}


 Now, we will upper bound the redundancy of $C_k$ for both the KT probability assignment and the Laplace probability assignment. By equation~\eqref{eqn:compress-length}, we have 
 \begin{align}
     \E[(\ell(C_k(A_n))]\le &{n' \choose 2}H(\BB_{12}) + 2^{k^2}\log\left(2e^2n^3\right)\nonumber\\ &+ (nk^2 H(A_{12})+4)+\E(2N_r\lceil\log n\rceil).\label{eq:expected_length_upper}
 \end{align}
 We will now analyze each of these four terms separately. 
 Firstly, we have
 \begin{align}
    &\binom{n'}{2}H(\mathbf{B}_{12})\nonumber\le \binom{n'}{2}k^2H(A_{1,k+1})\nonumber\\
    &=\binom{n'}{2}k^2h(f(n)\PQP)\nonumber\\
    &\stackrel{(a)}{=}\binom{n}{2}f(n)\bigg(\left(\log \frac{1}{f(n)}\right)\PQP+\PQP\log e\nonumber\\ &\;\;\;+\PQP\log \frac{1}{\PQP}+o(1)\bigg),\label{eq:o1term1}
 \end{align}
 where (a) follows from equation~\eqref{eq:hpqp-expand}.
Next, since $k\le \sqrt{\delta\log n}$ and $\delta<\epsilon$, we have 
\begin{equation}
    2^{k^2}\log\left(2e^2n^3\right)\le n^\delta\log(2e^2n^3)=o(n^2f(n)).\label{eq:o1term2}
\end{equation}
 Moreover, we have
 \begin{align}
     &nk^2H(A_{12})+4\le(n\delta\log n)H(A_{12})+4\nonumber\\ &=(n\delta\log n)O\left(f(n)\log\frac{1}{f(n)}\right)+4=o(n^2f(n)).\label{eq:o1term3}
 \end{align}
 Finally, by equation~\eqref{eqn:small-term-3}, we have
 \begin{equation}
     \E(2N_r\lceil\log n\rceil)=O(knf(n)\log n)=o(n^2f(n)).\label{eq:o1term4}
 \end{equation}
 Substituting~\eqref{eq:o1term1},~\eqref{eq:o1term2},~\eqref{eq:o1term3},~\eqref{eq:o1term4} into the upper bound~\eqref{eq:expected_length_upper} and combining with the lower bound for the entropy~\eqref{eq:entropylowerbound} gives
 \begin{align}
     &\E(\ell(C_k(A_n)))-H(A_n)\nonumber\\&\le \binom{n}{2}f(n)\bigg(\left(\log \frac{1}{f(n)}\right)\PQP+\PQP\log e\nonumber\\&\;\;\;+\PQP\log \frac{1}{\PQP}+o(1)\bigg)+o(n^2f(n)) \nonumber\\
     &\;\;\;-\binom{n}{2}f(n)\bigg(\log\left(\frac{1}{f(n)}\right)\PQP+\PQP\log e\nonumber\\ &\;\;\;+\PsQP+o(1) \bigg)\nonumber\\
     &=\binom{n}{2}f(n)\left(\PQP\log\left( \frac{1}{\PQP}\right)-\mathbf{p}^T\mathbf{Q^*p}\right)+o(n^2f(n)) \nonumber\\
     &=o(H(A_n)) \label{eq:GraphEntLimits}
 \end{align}
where~\eqref{eq:GraphEntLimits} follows since $H(A_n) = \Theta(n^2 f(n)\log(1/f(n)))$ by Lemma~\ref{lem:graphEntropy}.
 
 Now assume $f(n)=\Theta(1)$. Then $H(A_n)$ can be lower bounded as
 \begin{align}
     H(A_n)\ge H(A_n|X^n)&= \binom{n}{2}H(A_{12}|X_1,X_2)\nonumber\\ &=\binom{n}{2}\mathbf{p}^Th(f(n)\mathbf{Q})\mathbf{p}.\label{eq:theta1bound1}
\end{align}
Still, we can upper bound the expected length of the compressor $C_k$ using~\eqref{eq:expected_length_upper} and we will bound these four terms separately. We have 
\begin{align}
    &\binom{n'}{2}H(\mathbf{B}_{12})\nonumber\\&=\binom{n'}{2}(H(\mathbf{B}_{12}|X_1^{2k})+I(X_1^{2k};\mathbf{B}_{12}))\nonumber\\
    &=\binom{n'}{2}(k^2H(A_{1,k+1}|X_1,X_{k+1})+I(X_1^{2k};\mathbf{B}_{12}))\nonumber\\
    &\stackrel{(b)}{\le} \binom{n'}{2}(k^2\PhQP+H(X_1^{2k}))\nonumber\\
    &=\binom{n'}{2}(k^2\PhQP+2kH(\mathbf{p}))\nonumber\\
    &\le\frac{n(n-k)}{2}\PhQP+n(n'-1)H(\mathbf{p}),\label{eq:theta1bound2}
\end{align}
where (b) follows by Lemma~\ref{lem:technical},
\begin{equation}
    2^{k^2}\log (2e^2n^3)\le n^\delta\log (2e^2n^3)=o\left(\frac{n^2}{k}\right),\label{eq:theta1bound3}
\end{equation}

\begin{equation}
    nk^2H(A_{12})+4\le n\delta\log n H(A_{12})+4=o\left(\frac{n^2}{k}\right).\label{eq:theta1bound4}
\end{equation}
and
 \begin{align}
     &\E(2N_r\lceil\log n\rceil)\nonumber\\&=f(n)\PQP\left((n-\tilde{n})\tilde{n} + \binom{n-\tilde{n}}{2}\right)2\lceil\log n\rceil\nonumber\\&=O(kn\log n)=o\left(\frac{n^2}{k}\right).\label{eq:theta1bound5}
 \end{align}
 Combining the bounds~\eqref{eq:theta1bound1},~\eqref{eq:theta1bound2},~\eqref{eq:theta1bound3},~\eqref{eq:theta1bound4},~\eqref{eq:theta1bound5} gives
  \begin{align*}
    &\E(\ell(C_k(A_n))) - H(A_n)\\&\le \frac{n(n-k)}{2}\PhQP+n(n'-1)H(\mathbf{p})\\&\;\;\;-\frac{n(n-1)}{2}\PhQP+o\left(\frac{n^2}{k}\right)\\
    &=n(n'-1)H(\mathbf{p})+\frac{n(1-k)}{2}\PhQP+o\left(\frac{n^2}{k}\right)\\
    &\le H(\mathbf{p})\frac{n^2}{k}+o\left(\frac{n^2}{k}\right)\\
    &= o(H(A_n)),
\end{align*}
where the last line follows since $k = \omega(1)$ and $H(A_n) = \Theta(n^2)$ when $f(n) = \Theta(1)$.
\end{proof}

\section{Performance under the local weak convergence framework}\label{sec:BCentropy}
In this section, we analyze the performance of the proposed compressor under the local weak convergence framework introduced in~\cite{Bordenave_2014,Delgosha--Anatharam2019tit}. We focus on a subfamily of SBM in the $f(n) = 1/n$ regime, whose expected degree is identical for vertices in all communities, i.e., $\mathbf{Qp} = (\lambda,\ldots,\lambda)^T$.  We will show that the proposed compressor achieves the same performance in terms of BC entropy as the compressor proposed in~\cite{Delgosha--Anatharam2019tit}. 

We first introduce some basic definitions of rooted graphs in Subsection~\ref{subsec:definitions}. Then, we define the local weak convergence of graphs and derive the local weak convergence limit of the subfamily of stochastic block model in Subsection~\ref{subsec:LWC}. Finally, we review the definition of BC entropy in Subsection~\ref{subsec:BCentropy} and state the performance guarantee of our compression algorithm in Subsection~\ref{subsec:Achieve_BC}.

\subsection{Basic definitions on  rooted graphs}\label{subsec:definitions}
Let $G=(V,E)$ be a simple graph (undirected, unweighted, no self-loop), with $V$ a countable set of vertices and $E$ a countable set of edges. Let $u\stackrel{G}{\sim}v$ denote the connectivity of vertices $u$ and $v$ in $G$. $G$ is said to be \emph{locally finite} if, for all $v\in V$, the degree of $v$ in $G$ is finite. A rooted graph $(G,o)$ is a locally finite and connected graph $G=(V,E,o)$ with a distinguished vertex $o\in V$, called the root. Two rooted graphs $(G_1,o_1)=(V_1,E_1,o_1)$ and $(G_2,o_2)=(V_2,E_2,o_2)$ are \emph{isomorphic}, denoted as $(G_1,o_1)\simeq(G_2,o_2)$, if there exists a bijection $\pi:V_1\rightarrow V_2$ such that $\pi(o_1)=o_2$ and $u\stackrel{G_1}{\sim}v$ if and only if $ \pi(u)\stackrel{G_2}{\sim}\pi(v)$ for all $u,v \in V_1$. One can verify that this notion of isomorphism defines an equivalence relation on rooted graphs. Let $[G,o]$ denote the equivalence class corresponding to $(G,o)$. Let $\mathcal{G}^*$ denote the set of all locally finite and connected rooted graphs. For $(G,o)\in\mathcal{G}^*$ and $h\in \mathbb{N}$, we write $(G,o)_h$ for the truncated graph at depth $h$ of the graph $(G,o)$, in other words, the induced subgraph on the vertices such that their distance from the root is less than or equal to $h$. The equivalence classes $[G,o]_h$ follow a similar definition. Let $\mathcal{G}^*_h$ denote the set of all $[G,o]_h$. Now, we define a metric $d^*$ on $\mathcal{G}^*$. For any $[G_1,o_1]$ and $[G_2,o_2]$, let
\begin{align*}
\hat{h}:=&\sup\{h\in \mathbb{Z}^+:(G_1,o_1)_h\simeq(G_2,o_2)_h \\&\text{ for some } (G_1,o_1) \in [G_1,o_1], (G_2,o_2) \in [G_2,o_2]\}
\end{align*}
and define the metric $d^*$ as $$d^*([G_1,o_1],[G_2,o_2]):=\frac{1}{1+\hat{h}}.$$ As shown in~\cite{Aldous2007}, equipped with the metric defined above, $\mathcal{G}^*$ is a Polish space, i.e., a complete separable metric space. For this Polish space, let $\mathcal{P}(\mathcal{G}^*)$ denote the Borel probability measures on it. We say that a sequence of measures $\mu_n\in\mathcal{P(G^*)}$ \emph{converges weakly} to $\mu\in\mathcal{P}(\mathcal{G}^*)$, written as $\mu_n\rightsquigarrow\mu$, if for any bounded continuous function $f$ on $\mathcal{G}^*$, we have $\int fd\mu_n\rightarrow\int fd\mu$. It was shown in~\cite{billing} that $\mu_n\rightsquigarrow\mu$ if for any uniformly continuous and bounded functions $f$, we have $\int fd\mu_n\rightarrow\int fd\mu$. For $\mu\in \mathcal{P}(\mathcal{G}^*)$, $h\in\{0,1,2,\ldots\}$, and $[G,o]\in \mathcal{G}^*$, let $\mu_h$ denote the $h$-neighbourhood marginal of $\mu$
\begin{equation}
\mu_h([G,o])=\sum_{[G',o]\in \mathcal{G}^*:[G',o]_h=[G,o]}\mu([G',o]).\label{eqn:h-marginal-distn}
\end{equation}
For a locally finite graph $G=(V,E)$ and a vertex $v\in V$, let $G(v)$ denote the graph component in $G$ that is connected to $v$. By our previous definitions, $(G(v),v)$ denotes the rooted graph of the connected component of $v$ and the root is located at $v$ and $[G(v),v]$ denotes the equivalence class corresponding to $(G(v),v)$. Now, the \emph{rooted neighbourhood distribution} of $G$ is defined as the distribution of the rooted graph when the root is chosen uniformly at random over $V$
\begin{equation}
U(G):=\frac{1}{|V|}\sum_{v\in V}\delta_{[G(v),v]},\label{eqn:U(G)}
\end{equation}
where $\delta$ is the Dirac delta function.

For a sequence of graphs $\{G_n\}_{n=1}^\infty$, we say a probability distribution $\mu$ on $\mathcal{G}^*$ is the \emph{local weak limit} of $\{G_n\}_{n=1}^\infty$ if $U(G_n)\rightsquigarrow\mu$. Here we quote an interpretation of the term \emph{local} in this definition from~\cite{Delgosha--Anatharam2019tit}: An equivalent definition of the local weak limit $\mu$ is that, for any finite $h\geq 0$, if we randomly select a vertex $i$ and examine the $h$-hop local neighborhood centered at $i$, the resulting distribution on the neighborhood structure converges to $\mu$.
\subsection{Local weak convergence} \label{subsec:LWC}
In this section, we want to study the local weak limit of a sequence of graphs generated from stochastic block models. However, we notice that the definition for rooted neighbourhood distribution given in~\eqref{eqn:U(G)} is limited to deterministic graphs. Therefore, we first extend this definition to the case of random graphs. Suppose we have a random graph $G_n\sim \Lambda_n$, where $\Lambda_n$ is a probability distribution supported on graphs with $n$ vertices. Let $\mathrm{supp}(\Lambda_n)$ denote the support of $\Lambda_n$. Let $B$ denote a measurable set subset of $\mathcal{G}^*$. We define
\begin{align}
&U(G_n)(B)\nonumber\\&=\sum_{g\in\mathrm{supp}(\Lambda_n)}\P_{\Lambda_n}(G_n=g)\frac1n\sum_{i=1}^n\mathbbm{1}\{[g(i),i]\in B\},\label{eqn:U(G)-rand}
\end{align}
i.e., we first generate a random graph $G_n$ according to law $\Lambda_n$ and then pick a vertex $i$ uniformly at random independently from the graph generation process, and $U(G_n)(B)$ represents the probability that $[G_n(i),i]\in B$ from this process. By rearranging the order of the sums in equation~\eqref{eqn:U(G)-rand}, we have
\begin{align}
    &U(G_n)(B)\nonumber\\
    &=\frac1n\sum_{i=1}^n\sum_{g\in\mathrm{supp}(\Lambda_n)}\P_{\Lambda_n}(G_n=g)\mathbbm{1}\{[g(i),i]\in B\}\nonumber\\&=\frac1n\sum_{i=1}^n \P_{\Lambda_n}([G_n(i),i]\in B).\label{eqn:U(G)-rand-interpret}
\end{align}
Now let's consider the case when $\Lambda_n=\mathrm{SBM}(n,L,\pv,f(n)\Qv)$, and we have a graph $A_n\sim \Lambda_n$. By the exchangeability of stochastic block model shown in Lemma~\ref{lem:exchangeability}, we know that for any $B\subseteq\mathcal{G}^*$, $\P_{\Lambda_n}([A_n(i),i]\in B)$ is equal for all vertices $i\in [n]$. Therefore, by~\eqref{eqn:U(G)-rand-interpret}, we have 
\begin{equation}
    U(A_n)(B)=\P_{\Lambda_n}([A_n(\rho),\rho]\in B)
\end{equation}
for any vertex $\rho\in V(A_n)$, i.e., the rooted neighbourhood distribution $U(A_n)$ reduces to the neighbourhood distribution at arbitrary vertex $\rho$ in graph $A_n$ under the law $\Lambda_n$. We denote this neighbourhood distribution by $\mathrm{Nbr}(\rho)$. Next, we are going to study the limit property of $\mathrm{Nbr}_R(\rho)$, i.e. the $R$-neighbourhood marginal of $\mathrm{Nbr}(\rho)$, under certain assumptions on the parameters of the stochastic block model.
Since this distribution is identical for every vertex in $V(A_n)$, we will simply write $\mathrm{Nbr}_R$ for the $R$-neighbourhood distribution of any vertex.

To state the limiting distribution, we need to define the \emph{Galton--Watson tree} probability distribution on rooted trees $\mathrm{GWT}(\poi(\lambda))$. Let $\poi(\lambda)$ denote the Poisson distribution with mean $\lambda$. We take a vertex as the root and generate $Z^{(1)}\sim\poi(\lambda)$ as the number of children of the first generation. For the first generation, independent of $Z^{(1)}$, we generate $\xi^{(1)}_1,\ldots, \xi^{(1)}_{Z^{(1)}}$ i.i.d. according to $\poi(\lambda)$ as the number of children of each vertex in the first generation. Let $Z^{(2)}=\sum_{i=1}^{Z^{(1)}}\xi^{(1)}_i$ denote the total number of vertices in the second generation. In general, for the $j$th generation, $j = 1,2,\ldots$, generate the number of children for each vertex in the $j$th generation $\xi^{(j)}_1,\ldots,\xi^{(j)}_{Z^{(j)}}$ i.i.d. according to $\poi(\lambda)$, independent of all previous variables $\{ \xi^{(i-1)}_1, \ldots, \xi^{(i-1)}_{Z^{(i-1)}}, Z^{(i)}, \text{ for all } i \le j\}$. Let $Z^{(j+1)}= \sum_{k=1}^{Z^{(j)}}\xi^{(j)}_k$ denote the total number of vertices in the $j$th generation. In this way, we iteratively defined a measure on rooted trees.

Remember that the total variation distance between two probability measures $\mu_1$ and $\mu_2$ is defined as $\dtv(\mu_1,\mu_2):=\sup_{g:\mathcal{G}^*\rightarrow [-1,1]}\left|\int gd\mu_1-\int gd\mu_2\right|$. With the definitions above, we are ready to state the proposition that upper bounds the total variation distance between $\mathrm{Nbr}_R$ and a Galton--Watson tree.

 \begin{prop}\label{prop:neighbour_sim_GWT}
 Let $A_n \sim \mathrm{SBM}(n,L,\mathbf{p},\frac1n \mathbf{Q})$ with $\mathbf{Qp} = \lambda \mathbf{1}$ for some positive constant $\lambda$ and an all-$1$ vector $\mathbf{1}$ of dimension $L\times 1$.
 Let 
 $R=\floor*{\log\log\log n}$. Recall that $\mathrm{Nbr}_R$ and $\mathrm{GWT}(\poi(\lambda))_R$ denote the $R$-neighbourhood marginals of $\mathrm{Nbr}$ and $\mathrm{GWT}(\poi(\lambda))$ respectively, as defined in~\eqref{eqn:h-marginal-distn}. Then, the total variation distance satisfies as $n\rightarrow \infty$,
 $$d_\mathrm{TV}(\mathrm{Nbr}_R,\mathrm{GWT}(\poi(\lambda))_R)\rightarrow 0.$$
 
 \end{prop}

The proof of Proposition~\ref{prop:neighbour_sim_GWT} follows a similar idea as the proof of Proposition 2 in~\cite{Mossel--Neeman--Sly15}. The key idea essentially follows from the fact that $\mathrm{GWT}(\poi(\lambda))$ can be constructed from a sequence of Poisson random variables with mean $\lambda$, while $\mathrm{Nbr}_R$ can be constructed from a sequence of binomial random variables with approximately the same mean $\lambda$. The case when $R = 1$ follows essentially from the Poisson approximation of Binomial in the $1/n$ probability regime. This Proposition generalizes it to $R = O(\log n)$. 
 
 To upper bound the total variation distance between these two distributions, we first review the definition of coupling between two probability distributions. Let $\mu$ and $\nu$ be probability measures on the same measurable space $(S,\mathcal{S})$. A coupling of $\mu$ and $\nu$ is a probability measure $\gamma$ on the product space $(S\times S,\mathcal{S}\times \mathcal{S})$ such that the marginal of $\gamma$ coincide with $\mu$ and $\nu$, i.e.,
$$\gamma(A\times S)=\mu(A)\; and \; \gamma(S\times A)=\nu(A),\; \forall A\in\mathcal{S}.$$
It is known that for two probability distributions $\mu$ and $\nu$ on $(S,\mathcal{S})$ and a pair of random variables $(X,Y)\sim\gamma$, where $\gamma$ is a coupling of $\mu$ and $\nu$, it holds that
$$\dtv(\mu,\nu)\le \P(X\neq Y),$$
(see, e.g.,~\cite{Roch2020}).
Therefore, it suffices to construct a coupling $(G_R,T_R)$ such that the marginal distribution of $G_R$ is $\mathrm{Nbr}_R$, the marginal distribution of $T_R$ is $\mathrm{GWT}(\poi(\lambda))_R$ and $\P(G_R\neq T_R)\rightarrow 0$. To construct such a coupling, we will present an induction process in which each step holds with high probability. 

To formally state the proof of Proposition~\ref{prop:neighbour_sim_GWT}, we need a few more definitions. Let $V=V(A_n)$ denote the set of vertices in the graph $A_n$. For a vertex $\rho\in V$ and an integer $r\in[R]$, let $G_r(\rho)$ denote the (random) $r$-neighbourhood of $\rho$, i.e, $G_r(\rho)$ is the induced subgraph on the vertex set $\{v\in V:d(v,\rho)\le r\}$. We will simply write $G_r$ for $G_r(\rho)$. Let $V_r\triangleq V\setminus V(G_r)$ denote the set of vertices in $V$ that are not in $V(G_r)$. Let $\partial G_r\triangleq \{v\in V: d(v,\rho)=r\}$ denote the depth-$r$ neighbourhood of $\rho$. For a vertex $v\in\partial G_r$, let $Y_v$ denote the number of neighbours $v$ has in $V_r$. Let $T_R$ be a random tree generated from $ \mathrm{GWT}(\poi(\lambda))_R$, i.e., $T_R\sim\mathrm{GWT}(\poi(\lambda))_R.$ For a vertex $u\in V(T_R)$, let $Z_u$ denote the number of children of $u$ in $T_R$. In order to couple $G_R$ with $T_R$ such that $\P(G_R\neq T_R)\rightarrow 0$, a necessary condition is that $G_R$ is a tree with high probability. Now, we define two sets of events to guarantee that $G_R$ is a tree. For any $r\in[R]$, let $C_r$ denote the event that no vertex in $V_{r-1}$ has more than one neighbour in $G_{r-1}$ and let $D_r$ denote the event that there are no edges within $\partial G_r$. Clearly, if $C_r$ and $D_r$ holds for every $r\in[R]$, then $G_R$ is indeed a tree. With these definitions, we state a lemma that introduces the induction step.
\begin{lem}\label{lem:induction}
If
\begin{enumerate}
    \item $G_{r-1}=T_{r-1}$;
    \item $Y_u=Z_u$ for every $u\in \partial G_{r-1};$
    \item $C_r$ and $D_r$ hold
\end{enumerate}
then $G_r=T_r$.
\end{lem}

\begin{proof}[Proof of Lemma~\ref{lem:induction}]
First of all, events $C_r$ and $D_r$ guarantee that $G_r$ is a tree. Moreover, $G_{r-1}=T_{r-1}$ and $Y_u=Z_u$ for every vertex in the last generation of $G_{r-1}$ and $T_{r-1}$, so we have $G_r=T_r$.
\end{proof}
Next, we define a set of auxiliary events in order to complete the proof of our induction step. For any $0\le r\le R$, we define $E_r$ to be the event
$$E_r=\{|\partial G_s|\le 2^s Q_\mathrm{max}^s\log n,\; \forall s\le r+1\}.$$
The utility of this auxiliary event is shown in the following two lemmas.

\begin{lem}\label{lem:bound_num_vertex}
For all $r\le R$ and some constant $c>0$,
$$\P(E_r|E_{r-1})\ge 1-n^{-c}.$$
Moreover, $|G_r|=O(\log^2 n)$ when $E_{r-1}$ holds true.
\end{lem}
\begin{proof}[Proof of Lemma~\ref{lem:bound_num_vertex}]
First of all, remember that for two random variables $X_1$ and $X_2$, we say that $X_1$ stochastically dominates $X_2$ if $\P(X_1\ge x)\ge \P(X_2\ge x)$ for any $x$.
By our definition, $Y_v$ is stochastically dominated by $\mathrm{Binom}(n,Q_\mathrm{max}/n)$ for any $v$. On $E_{r-1}$, $|\partial G_r|\le 2^rQ_\mathrm{max}^r\log n$ and so $|\partial G_{r+1}|$ is stochastically dominated by
$$Z\sim\mathrm{Binom}(2^r\Qm^rn\log n,\frac{\Qm}{n}).$$
Thus,
\begin{align*}
    \P(E_r^c|E_{r-1})&=\P(|\partial G_{r+1}|>2^{r+1}\Qm^{r+1}\log n |E_{r-1})\\
    &\le\P(Z\ge 2\E Z)\stackrel{(a)}{\le}\left(\frac{e}{4}\right)^{\E Z},
\end{align*}
where (a) follows by a multiplicative version of Chernoff's inequality~\cite{Chernoff1952}
$$\P(X>(1+\delta)\E X)\le\left(\frac{e^\delta}{(1+\delta)^{1+\delta}}\right)^{\E X}$$
for binomial random variable $X$ and any $\delta>0$.
Because $R=\floor*{\log \log \log n}$, we have
$$\E Z=2^r\Qm^{r+1}\log n=\Omega( \log ^{c_1} n),$$
for some positive constant $c_1$,
which proves the first part of the Lemma.\newline
For the second part, on $E_{r-1}$
\begin{align*}
|G_r|&=\sum_{r=1}^R|\partial G_r|\\&\le\sum_{r=1}^R2^r\Qm ^r\log n\le (2\Qm)^{R+1}\log n=O(\log^2 n),
\end{align*}
since $R=\floor*{\log \log \log n} $.
\end{proof}

\begin{lem}\label{lem:is_tree}
For any $r$, 
$$\P(C_r|E_{r-1})\ge 1-O(n^{-3/4})$$
$$\P(D_r|E_{r-1})\ge 1-O(n^{-3/4}).$$
\end{lem}
\begin{proof}
For the first claim, fix $u,v\in \partial G_r$. For any $w\in V_r$, the probability that $(u,w)$ and $(v,w)$ both appear is $O(n^{-2})$. Now, $|V_r|\le n$ and Lemma~\ref{lem:bound_num_vertex} implies that $|\partial G_r|^2=O(\log^4 n)$. Hence the result follows from the union bound over all triples $u,v,w$.
For the second claim, the probability of an edge between any particular pair of vertices $u,v\in\partial G_r$ is $O(n^{-1})$. Lemma~\ref{lem:bound_num_vertex} implies that $|\partial G_r|^2=O(\log^4 n)$ and the result follows from a union bound over all the pairs $u,v$.
\end{proof}

Moreover, we state a lemma to bound the total variation distance between the Binomial distribution and the Poisson distribution. 
\begin{lem}[Lemma~5 in~\cite{Mossel--Neeman--Sly15}]
\label{lem:Binom_Pois_TV}
If $m$ and $n$ are positive integers and $c>0$ is a positive constant, then
$$d_\mathrm{TV}(\mathrm{Binom}(m,c/n),\poi(c))=O\left(\frac{\max\{1,|m-n|\}}{n}\right).$$
\end{lem}
Now, we are ready to prove Proposition~\ref{prop:neighbour_sim_GWT}.
\begin{proof}[{Proof of Proposition~\ref{prop:neighbour_sim_GWT}}]
For any random variable $X$, let $\mathrm{dist}(X)$ denote the probability distribution of $X$. For a vertex $v\in V$, let $X_v$ denote its community label. Fix $r$ and suppose $E_{r-1}$ holds and $T_{r-1}=G_{r-1}$. Let $V^{(1)},\ldots, V^{(L)}$ denote the set of vertices in $V$ with community label $1,\ldots, L$ respectively and similarly, let $V_{r-1}^{(1)},\ldots, V_{r-1}^{(L)}$ denote the set of vertices in $V_{r-1}$ with community label $1,\ldots, L$ respectively. Let us first provide a high probability bound for $|V^{(1)}|,\ldots, |V^{(L)}|$. 
For each $i\in [L]$, by Hoeffding's inequality~\cite{Hoeffding1963}, we have
$$\P\left(||V^{(i)}|-np_i|\ge n^{3/4}\right)\le 2e^{-\Omega(n^{1/2})}.$$
It follows from the union bound that
\[
\P\left(\forall i\in[L], ||V^{(i)}|-np_i| \le n^{3/4}\right) \ge 1 - 2Le^{-\Omega(n^{1/2})}.
\]
Combining with Lemma~\ref{lem:bound_num_vertex}, conditioned on event $E_{r-1}$ and the event that $\{\forall i\in[L], ||V^{(i)}|-np_i| \le n^{3/4}\}$, the number $|V_{r-1}^{(i)}|$ of vertices in $V_{r-1}$ with community label $i$ can be upper and lower bounded as
\begin{align*}
np_i+n^{3/4}&\ge |V^{(i)}|\ge |V^{(i)}_{r-1}|\ge |V^{(i)}|-|G_{r-1}|\\&\ge np_i-n^{3/4}-O(\log^2 n)
\end{align*}
for any $i\in[L]$.

Next, let us analyze the distribution of $Y_v$, i.e., the number of neighbours in $V_{r-1}$ for a vertex $v \in \partial G_{r-1}$. Let $Y_v^{(1)},\ldots, Y_v^{(L)}$ denote the number of neighbours $v$ has in $V_{r-1}$ with community label $1,\dots, L$ respectively. Then we have $Y_v=\sum_{i=1}^LY_v^{(i)}$ and $Y_v^{(1)},\ldots, Y_v^{(L)}$ are independent of each other given the community label $X_v$. We know that conditional distribution of $Y_v^{(i)}$ given $X_v=j$ is $\mathrm{Binom}(|V_{r-1}^{(i)}|,Q_{ji}/n)$. By Lemma~\ref{lem:Binom_Pois_TV}, we have $d_\mathrm{TV}(\mathrm{Binom}(|V_{r-1}^{(i)}|,Q_{ji}/n),\poi(p_iQ_{ji}))=O(n^{-1/4})$. For any $i\in[L]$, let $Z^{(i)}$ denote a Poisson random variable with mean $p_iQ_{ji}$, i.e., $Z^{(i)}\sim \poi(p_iQ_{ji})$. Therefore, conditioning on any $X_v=j$ for any $i\in[L]$ we can couple $Y_v^{(i)}$ with $Z^{(i)}$ such that $\P(Y_v^{(i)}\neq Z^{(i)})=O(n^{-1/4})$. By a union bound over $L$ communities, we can couple the $L$-tuples $(Y_v^{(1)},\ldots,Y_v^{(L)})$ and $(Z^{(1)},\ldots,Z^{(L)})$ such that $\P((Y_v^{(1)},\ldots,Y_v^{(L)})\neq(Z^{(1)},\ldots,Z^{(L)}))=O(n^{-1/4})$. By Poisson superposition, it follows that $d_\mathrm{TV}(\mathrm{dist}(Y_v|X_v=j),\poi(\sum_{i=1}^Lp_iQ_{ji}))=O(n^{-1/4})$. Recall our assumption that $\sum_{i=1}^Lp_iQ_{ji}=\lambda$ for any $j\in[L]$ and that $Z_v\sim\poi(\lambda)$. Moreover, $Y_v$ is conditionally independent of $G_{r-1}$ given its community label $X_v$. Therefore, we have
\begin{align*}
&d_\mathrm{TV}(\mathrm{dist}(Y_v|G_{r-1}=g_{r-1}),\mathrm{dist}(Z_v)) \\&= d_\mathrm{TV}\left(\mathrm{dist}(Y_v|X_v=j),\poi\left(\sum_{i=1}^Lp_iQ_{ji}\right)\right)=O(n^{-1/4})
\end{align*}
for any possible realization $g_{r-1}$.
Thus, we can couple $Y_v$ with $Z_v$ such that $\P(Y_v\neq Z_v)=O(n^{-1/4})$. Then by a union bound over all vertices in $\partial G_{r-1}$, we have
\begin{align*}
    \P(\exists v \in \partial G_{r-1}: Y_v \neq Z_v) &\le |\partial G_{r-1}|\P(Y_v \neq Z_v)\\
    &= O(\log^2 n)O(n^{-1/4})\\
    &=O(n^{-1/5}).
\end{align*}
The argument above shows that we can find a coupling such that with probability at least $1-O(n^{-1/5})$, $Y_u=Z_u$ for any $u\in\partial G_{r-1}$. Moreover, Lemmas~\ref{lem:bound_num_vertex} and~\ref{lem:is_tree} imply that $C_r$, $D_r$ and $E_r$ hold simultaneously with probability at least $1-n^{-c}-O(n^{-3/4})$. Putting these all together, we see that the hypothesis of Lemma~\ref{lem:induction} holds with probability at least $1-O(n^{-\mathrm{min}(c,1/5)})$. Thus,
$$\P(G_{r}=T_{r},E_r|G_{r-1}=T_{r-1},E_{r-1})\ge 1-O(n^{-\mathrm{min}(c,1/5)}).$$
For the base case, we have $\P(E_0)\rightarrow 1$ as $n\rightarrow \infty$ by the multiplicative Chernoff's inequality, and we can certainly couple $G_0$ with $T_0$ since they are both a single vertex. Therefore, we can apply a union bound over $r=1,\ldots,R$, 
\begin{align*}
    &\P(G_R\neq T_R)\\&=\P(\exists r\in[R]\; \mathrm{s.t.}\; G_{r}\neq T_{r}\text{ or }E_r^c|G_{r-1}=T_{r-1},E_{r-1})\\
    &\le R\cdot O(n^{-\mathrm{min}(c,1/5)})\\
    &=\Theta(\log\log\log n) O(n^{-\mathrm{min}(c,1/5)})\\ &\rightarrow0
\end{align*}
as $n\rightarrow \infty$, i.e.,
we can couple $G_R$ and $T_R$ such that $G_R=T_R$ with high probability. Therefore, we have $d_\mathrm{TV}(\mathrm{Nbr}_R,\mathrm{GWT}(\poi(\lambda))_R)\rightarrow 0$ as $n\rightarrow\infty$.
\end{proof}

With Proposition~\ref{prop:neighbour_sim_GWT}, we are ready to establish the local weak convergence of the stochastic block model.

\begin{prop}[Local weak limit of sparse SBMs]\label{LWC_SBM}
Consider a sparse stochastic block model $\mathrm{SBM}(n,L,\mathbf{p},\frac{1}{n}\mathbf{Q})$ with $\mathbf{Qp} = \lambda \mathbf{1}$ for some positive constant $\lambda$, where $\mathbf{1}$ denote the all-$1$ vector of dimension $L\times 1$. Let $\{A_n\}_{n=1}^\infty$ denote a sequence of random graphs with $A_n\sim \mathrm{SBM}(n,L,\mathbf{p},\frac{1}{n}\mathbf{Q})$. Let $U(A_n)$ be the rooted neighbourhood distribution of $A_n$.
Then, the local weak limit of the sequence of graphs $\{A_n\}_{n=1}^\infty$ is given by $\mathrm{GWT}(\poi(\lambda))$,\footnote{To the best of our knowledge, the result of Proposition~\ref{LWC_SBM} was novel when it was first presented in~\cite{Bhatt--Wang--Wang--Wang2020}. Later in~\cite{hofstad2023}, van der Hofstad independently derived a more general result which includes Proposition~\ref{LWC_SBM} as a special case. } i.e.,
$$U(A_n)\rightsquigarrow \mathrm{GWT}(\poi(\lambda)).$$
\end{prop}
\begin{rem}
When $Q_{i,j} = c$ for all $i,j \in [L]$, the stochastic block model recovers the well-known local weak convergence result on the Erd\H{o}s--R\'{e}nyi model (see, e.g.,~\cite[Theorem 3.12]{Bordenave2016}).
\end{rem}
\begin{proof}[{\bf Proof of Proposition~\ref{LWC_SBM}}]
We want to show that for any uniformly continuous and bounded function $f$, 
$$\left|\int fdU(A_n)-\int f d\mathrm{GWT}(\poi(\lambda))\right|\rightarrow 0$$
as $n\rightarrow \infty$.  
Since $f$ is a uniformly continuous function on $\mathcal{G}^*$, for every $\epsilon >0$ there exists $\delta >0$ such that, for any pair of rooted graphs $[G_1,o_1]$ and $[G_2,o_2]\in \mathcal{G^*}$ with $d^*([G_1,o_1],[G_2,o_2])<\delta$ we have $|f(G_1,o_1)-f(G_2,o_2)|< \epsilon$. Recall that $d^*([G_1,o_1],[G_2,o_2]):=\frac{1}{1+\hat{h}}$, where $\hat{h}$ denotes the maximum layers of matching between $[G_1,o_1]$ and $[G_2,o_2]$. Therefore, as long as $h>\frac{1}{\delta}-1$, we have $|f((G,o)_h)-f(G,o)|< \epsilon$. It follows that $|f([i,o])-f([g,o])|<\epsilon$, if $[i,o]_h=[g,o]$. Let $\mu\in\mathcal{P}(\mathcal{G}^*)$ and assume $h>\frac{1}{\delta}-1$. We have
{\allowdisplaybreaks
\begin{align*}
    &\left|\int fd\mu_h-\int fd\mu\right|\nonumber\\&=\left|\sum_{[g,o]\in \mathcal{G}^*_h}f([g,o])\mu_h([g,o])-\sum_{[i,o]\in \mathcal{G}^*}f([i,o])\mu([i,o])\right|\\
    &\le \sum_{[g,o]\in\mathcal{G}^*_h}\Bigg|f([g,o])\mu_h([g,o])
    \nonumber\\&\hspace{5em}-\sum_{[i,o]\in\mathcal{G}^*:[i,o]_h=[g,o]}f([i,o])\mu([i,o])\Bigg|\\
    &\stackrel{(a)}{=}\sum_{[g,o]\in\mathcal{G}^*_h}\left|\sum_{[i,o]\in\mathcal{G}^*:[i,o]_h=[g,o]}(f([g,o])-f([i,o]))\mu([i,o])\right|\\
    &\le\sum_{[g,o]\in\mathcal{G}^*_h}\sum_{[i,o]\in\mathcal{G}^*:[i,o]_h=[g,o]}\left|f([g,o])-f([i,o])\right|\mu([i,o])\\
    &\le\sum_{[g,o]\in\mathcal{G}^*_h}\sum_{[i,o]\in\mathcal{G}^*:[i,o]_h=[g,o]}\epsilon\mu([i,o])=\epsilon,
\end{align*}}%
where (a) follows since $\mu_h([g,o])=\sum_{[i,o]\in \mathcal{G}^*:[i,o]_h=[g,o]}\mu([i,o])$.
Therefore, $|\int fdU(A_n)_h-\int fdU(A_n)|<\epsilon$ and $|\int fd\mathrm{GWT}(\poi(\lambda))_h-\int fd\mathrm{GWT}(\poi(\lambda))|<\epsilon$. By Proposition~\ref{prop:neighbour_sim_GWT}, there exists $n_0$ such that if $n\ge n_0$ and 
$\floor*{\log\log\log n}\ge R$, we have $\dtv(\mathrm{GWT}(\poi(\lambda))_R,U(A_n)_R)<\epsilon$. Since $f$ is a bounded function, we have $|\int f d\mathrm{GWT}(\poi(\lambda))_R-\int fdU(A_n)_R|<\epsilon$, as long as $n$ is large enough. Therefore, if we take $n$ large enough such that 
$\floor*{\log\log\log n}>\frac{1}{\delta}-1$ and $|\int f d\mathrm{GWT}(\poi(\lambda))_h-\int fdU(A_n)_h|<\epsilon$, we have 
\begin{align*}
    &\left|\int f dU(A_n)-\int fd\mathrm{GWT}(\poi(\lambda))\right|\\&\le \left|\int fdU(A_n)_h-\int fdU(A_n)\right|\\
    &\;\;\;+\left|\int fd\mathrm{GWT}(\poi(\lambda))_h-\int fd\mathrm{GWT}(\poi(\lambda))\right|\\
    &\;\;\;+\left|\int f d\mathrm{GWT}(\poi(\lambda))_h-\int fdU(A_n)_h\right|\\
    &<3\epsilon,
\end{align*}
which completes the proof. 
\end{proof}

\subsection{BC entropy}\label{subsec:BCentropy}
In this section, we review the notion of BC entropy introduced in~\cite{Bordenave_2014}, which is shown to be the fundamental limit of universal lossless compression for certain graph family~\cite{Delgosha--Anatharam2017}. 

For a Polish space $\Omega$, let $\mathcal{P}(\Omega)$ denote the set of all Borel probability measures on $\Omega$. Let $A$ be a Borel set in $\Omega$, we define the \emph{$\epsilon$-extension of A}, denoted $A^\epsilon$, as the union of the open balls with radius $\epsilon$ centered around the points in $A$. For two probability measures $\mu$ and $\nu$ in $\mathcal{P}(\Omega)$, we define the \emph{L\'{e}vy--Prokhorov distance} $d_\mathrm{LP}(\mu,\nu):=\inf\{\epsilon>0:\mu(A)\le\nu(A^\epsilon)+\epsilon \text{ and } \nu(A)\le\mu(A^\epsilon)+\epsilon,\forall A\in\mathcal{B}(\Omega)\}$, where $\mathcal{B}(\Omega)$ denotes the Borel sigma algebra of $\Omega$. 
Let $\rho\in \mathcal{P}(\mathcal{G}^*)$. Let $d$ be the expected number of neighbours of root under the law $\rho$ and let a sequence $m=m(n)$ such that $m/n\rightarrow d/2$, as $n\rightarrow\infty$. Define
$\Gnm$ to be the set of graphs with $n$ vertices and $m$ edges. For $\epsilon>0$, define 
$$\Gnm(\rho,\epsilon)=\{G\in \Gnm:U(G)\in B(\rho,\epsilon)\},$$
where $B(\rho,\epsilon)$ denotes the open ball with radius $\epsilon$ around $\rho$ with respect to Lévy–Prokhorov metric. Now, we define the \emph{$\epsilon$-upper BC entropy} of $\rho$ as
$$\upperBC(\rho,\epsilon)=\limsup\limits_{n\rightarrow\infty}\frac{\log |\Gnm(\rho,\epsilon)|-m\log n}{n}$$
and define the \emph{upper BC entropy} of $\rho$ as
$$\upperBC(\rho)=\lim_{\epsilon\rightarrow 0}\upperBC(\rho,\epsilon).$$
Similarly we define the \emph{$\epsilon$-lower BC entropy} $\lowerBC(\rho,\epsilon)$ and \emph{lower BC entropy} $\lowerBC(\rho)$ with $\limsup$ replaced by $\liminf$ in above definitions. If $\rho$ is such that $\upperBC(\rho)=\lowerBC(\rho)$, then this common limit is called the \emph{BC entropy} of $\rho$
$$\Sigma(\rho):=\upperBC(\rho)=\lowerBC(\rho).$$

The following lemma states the BC entropy of the Galton--Waston tree distribution. 
\begin{lem}[Corollary 1.4 of\cite{Bordenave_2014}]\label{SBM_BCentropy}
The BC entropy of the Galton--Watson tree distribution $\mathrm{GWT}(\poi(\lambda))$ is given by
$$\Sigma\left(\mathrm{GWT}(\poi(\lambda))\right)=\frac{\lambda}{2}\log \frac{e}{\lambda} \quad \mathrm{bits}.$$
\end{lem}

\subsection{Achieving BC entropy in the sparse regime}\label{subsec:Achieve_BC}
With the Lemma above, we can give a performance guarantee of our algorithm corresponding to the BC entropy. It is a Theorem analogous to Proposition 1 in\cite{Delgosha--Anatharam2017}.
\begin{thm}\label{thm:Achieve_BC}
Let $A_n\sim\mathrm{SBM}\left(n,L,\mathbf{p},\frac{1}{n}\mathbf{Q}\right)$ with $\mathbf{Qp} = \lambda \mathbf{1}$ for some positive constant $\lambda$ and an all-$1$ vector $\mathbf{1}$ of dimension $L\times 1$. Let $m= \binom{n}{2}\frac{\lambda}{n}$ be the expected number of edges in the model. Then, our compression algorithm achieves the BC entropy of the local weak limit of stochastic block models in the sense that 
$$\limsup_{n \to \infty}\frac{\E[\ell(C_k(A_n))]-m\log n}{n}\le \Sigma\left(\mathrm{GWT}(\poi(\lambda))\right).$$
\end{thm}
\begin{proof}
By our proof of Proposition~\ref{prop:redundancyStrong}, we have
\begin{align*}
\E[\ell(C_k(A_n))]\le &\binom{n'}{2}H(\mathbf{B}_{12})+2^{k^2}\log(2en^3)\\&+nk^2_nH(A_{12})+4+\E(2N_r\ceil*{\log n}).
\end{align*}
Notice that 
{\allowdisplaybreaks
\begin{align*}
    \binom{n'}{2}H(\mathbf{B}_{12})&\le \binom{n'}{2}k^2H(A_{12})\\
    &=\binom{n'}{2}k^2h(\lambda/n)\\
    &\stackrel{(a)}{=}\binom{n'}{2}k^2\left(\frac{1}{n}\lambda\log \frac{ne}{\lambda}+o\left(\frac{1}{n}\right)\right)\\
    &\stackrel{(b)}{\sim} \binom{n}{2}\left(\frac{1}{n}\lambda\log n +\frac{1}{n}\lambda\log \frac{e}{\lambda}+o\left(\frac{1}{n}\right)\right)\\
    &=\binom{n}{2}\frac{1}{n}\lambda\log n+\frac{\lambda\log e-\lambda\log \lambda}{2}n+o(n)\\
    &\stackrel{(c)}{=}m\log n+n\Sigma\left(\mathrm{GWT}(\poi(\lambda))\right)+o(n)
\end{align*}}%
where (a) follows since $h(p)=p\log \frac{e}{p}-\frac{\log e}{2}p^2+o(p^2)$, (b) follows since $n'k=n$ and (c) follows from Lemma~\ref{SBM_BCentropy}.
Then it suffices to that the remaining terms in the upper bound of $\E[\ell(C_k(A_n))]$ are all $o(n)$. Indeed we have 
\begin{align*}
    2^{k^2}\log (2en^3)\le 2^{\delta\log n}\log (2en^3)=n^\delta \log (2en^3)=o(n)
\end{align*}
since $\delta<1$, and 
\begin{align*}
    nk^2_nH(A_{12})+4&=nk^2_nh(\lambda/n)+4\\
    &=nk^2_n\left(\frac{1}{n}\lambda\log \frac{ne}{\lambda}+o\left(\frac{1}{n}\right)\right)+4\\
    &\le n\delta \log n \left(\frac{1}{n}\lambda\log \frac{ne}{\lambda}+o\left(\frac{1}{n}\right)\right)+4\\
    &=\delta\log n \left(\lambda\log \frac{ne}{\lambda}\right)+o\left(\log n\right)\\
    &=o(n),
\end{align*}
and 
\begin{align*}
    &\E(2N_r\lceil\log n\rceil)\\&=\frac1n\PQP\left((n-\tilde{n})\tilde{n} + \binom{n-\tilde{n}}{2}\right)2\lceil\log n\rceil\\&=O(k\log n)=o(n).
\end{align*}
\end{proof}

Theorem \ref{thm:Achieve_BC} shows that in the sparse regime $f(n) = \frac{1}{n}$, our compressor achieves the BC entropy of the Galton--Watson tree, which is the local weak convergence limit of the stochastic block model.

\section{Comparison to alternative methods}\label{sec:C1discussion}
In this section, we consider three alternative compression methods which are natural attempts at designing universal graph compressors for stochastic block models. We demonstrate why they are not applicable to our problem or not universal in the same sense as the proposed algorithm.

One intuitive attempt is to convert the entries in the adjacency matrix into a one-dimensional sequence according to some ordering and then apply a universal compressor to the sequence.
Let us take a closer look at the correlation among entries in the adjacency matrix and explain why existing universal sequence compressors developed for stationary processes may not be immediately applicable for certain orderings of the entries. Compressing $A_n$ entails compressing $A_{12},\ldots,A_{1,n},A_{23},\ldots,A_{n-1,n}$,
i.e. the bits in the upper triangle of $A_n$. Clearly, these are not independent (because of the dependency through $X_1^n$) so one cannot use any of the compressors universal for the class of iid processes to compress $A_n$. 
Moreover, we can show that some of the most natural orders of listing these ${n \choose 2}$ bits result in a sequence that is non-stationary.
    For example, consider listing the bits in the upper triangle row-wise (i.e. first listing the bits in the first row, followed by the bits in the second and so on, ending with $A_{n-1,n}$) resulting in the sequence $A_{12},\ldots,A_{1,n},A_{23},\ldots,A_{2,n},\ldots,A_{n-1,n}$. This
 can be seen to be nonstationary for example by considering the case when $n = 4, L = 2, Q_{11} = Q_{12} = 1,  Q_{12} = 0$. In this case the horizontal ordering is $A_{12},A_{13},A_{14},A_{23},A_{24},A_{34}$
    and we have $P(A_{12} = 1 ,A_{13} = 0,A_{14} = 1) > 0$ but $P(A_{23}=1 ,A_{24} = 0,A_{34} = 1) = 0$. Similar counterexamples can be given for establishing the nonstationarity of column-by-column ordering and diagonal-by-diagonal ordering. Therefore, it is unclear whether there exists an ordering that converts entries in the adjacency matrix into a stationary sequence.

Another potential approach is to first recover the community labels of all vertices, which allows one to rearrange the entries in the adjacency matrix into i.i.d. groups, and then apply a universal sequence compressor to each group of entries. There are two main issues with this approach. First, even when the exact recovery of labels is possible, to the best of our knowledge, all algorithms for community detection require knowledge of the graph parameters such as the number of communities and/or the edge probabilities~\cite{Abbe17}. Second, it is known that exact recovery is information-theoretically impossible when the edge probabilities are $o\left(\frac{\log n}{n}\right)$ but our compressor is universal even in this regime. This implies that universal compression is a fundamentally easier task than community detection for stochastic block models. 

Another approach is based on run-length coding for the adjacency matrix $A_n$, which stores the run-length of one of the symbols for the entire sequence, while storing the other symbol uncoded. For example, the sequence \texttt{AAAABAABBAA}, when encoded using the run length of \texttt{A}s, would be stored as \texttt{(4,B,2,B,0,B,2)} (so that the run length of \texttt{A}s between each \texttt{B} are stored). When applied to our problem, the expected length of this method can be shown to be upper bounded by 
${\bf p}^TQ{\bf p}\binom{n}{2}f(n)[\log(1/f(n)Q_\mathrm{min}) + O(\log\log n)]$ (by suitably modifying~\cite{Schilling1990}), where $Q_\mathrm{min} = \min_{i,j}Q_{i,j}$. Thus it fails to achieve the leading term in entropy when $f(n) = \Omega(1/\log n)$, since $p_\mathrm{min} = \Theta(f(n))$ and the term with $\log\log n$ contributes to the leading term. The analysis also fails in the case when there are zero entries in $W$ (since $p_\mathrm{min} = 0$), which is a reasonable choice of parameters for the SBM. In contrast, the analysis for $C_k$ is valid even when there are zero entries in $W$ as long as $\max_{ij}Q_{ij} = \Theta(1)$.

\section{Concluding Remarks}
\label{sec:conclusion}

In this paper, we constructed a compressor  that is universal for the class of stochastic block models with connection probabilities ranging from $\Omega\left(\frac{1}{n^{2-\eps}}\right)$ to $O(1)$.  There are many intriguing open problems that remain. Firstly, even though our algorithm takes polynomial time, $O(n^2)$ may be restrictive when $n$ is too large, and a lower complexity method is desirable. Secondly, a tight characterization of the minimax redundancy in this problem is of interest. Finally, we have considered the problem of universal compression in the stochastic setting, where we assumed that our graph is generated from a random graph model. This leads naturally to the question of the individual graph setting, where the redundancy must be established relative to a class of compressors. Each of these are promising avenues for further study.

\section*{Acknowledgment}
L. Wang would like to thank Emmanuel Abbe and Tsachy Weissman for stimulating discussions in the initial phase of the work. She is grateful to Young-Han Kim and Ofer Shayevitz for their interest and encouragement in this result. C. Wang would like to thank Richard Peng for his suggestion in writing. The authors would like to thank the associate editor and the anonymous reviewers for their invaluable contributions and constructive suggestions, which have significantly enhanced the quality and clarity of this paper.
\appendix
\subsection{Proof of Lemma~\ref{lem:KTlength}}\label{sec:lem-kt-proof}
Let $p_1=N_1/N, p_2=N_2/N,\cdots,p_m=N_m/N.$ Notice that $\sum_{i=1}^m p_i=1$, so $(p_1,\cdots p_m)$ can be viewed as a probability distribution. And the entropy of this distribution is $H(p_1, \cdots p_m)=\sum_{i=1}^m -p_i\log{p_i}.$ Firstly we consider the case when $N_1,N_2\cdots N_m$ are all positive and none of them equal to N.
By Stirling's approximation for factorial $\sqrt{2\pi n}(\frac{n}{e})^n e^{1/(12n+1)}\le n!\le \sqrt{2\pi n}(\frac{n}{e})^n e^{1/12n}$, we can bound
{\allowdisplaybreaks
\begin{align*} 
&\binom{N}{N_1,N_2\cdots N_m}\\&\ge \frac{\sqrt{2\pi N}N^N \exp\left(\frac{1}{12N+1}-\frac{1}{12N_1}-\frac{1}{12N_2}-\cdots -\frac{1}{12N_m}\right)}{(2\pi)^{m/2}(N_1N_2\cdots N_m)^{1/2}N_1^{N_1}N_2^{N_2}\cdots N_m^{N_m}}\\
&=\frac{\exp\left(\frac{1}{12N+1}-\frac{1}{12N_1}-\frac{1}{12N_2}-\cdots -\frac{1}{12N_m}\right)}{(2\pi)^{\frac{m-1}{2}}(p_1p_2\cdots p_m)^{1/2}N^{\frac{m-1}{2}}2^{-NH(p_1,p_2,\cdots,p_m)}}.
\end{align*}}%
Similarly, we have
{\allowdisplaybreaks
\begin{align*}
&\binom{2N}{2N_1,2N_2\cdots 2N_m}\\&\le \frac{\exp\left(\frac{1}{24N}-\frac{1}{24N_1+1}-\frac{1}{24N_2+1}-\cdots -\frac{1}{24N_m+1}\right)}{(2\pi)^{\frac{m-1}{2}}2^{\frac{m-1}{2}}(p_1p_2\cdots p_m)^{1/2}N^{\frac{m-1}{2}}2^{-2NH(p_1\cdots p_m)}}\cdot
\end{align*}}%
Consider the function 
\begin{align*}
&f(N_1,N_2,\cdots,N_m)\\&=\tfrac{1}{12N+1}-\tfrac{1}{24N}+(\tfrac{1}{24N_1+1}-\tfrac{1}{12N_1})+(\tfrac{1}{24N_2+1}-\tfrac{1}{12N_2})\\&\;\;\;+\cdots +(\tfrac{1}{24N_m+1}-\tfrac{1}{12N_m})
\end{align*}
and the function
$$g(n)=\frac{1}{24n+1}-\frac{1}{12n},$$
where $n$ is a positive integer.
Function $g(n)$ is minimized with $n=1$ and $\min g(n)=1/25-1/12$ and we can bound function $f(N_1,N_2,\cdots,N_m)\ge \frac{1}{12N+1}-\frac{1}{24N}+(1/25-1/12)m$. 
Finally we are ready to prove the lemma. 
\begin{align*} 
&\frac{\binom{N}{N_1,N_2\cdots N_m}\theta_1^{N_1}\cdots\theta_m^{N_m}}{\binom{2N}{2N_1,2N_2\cdots 2N_m}\theta_1^{2N_1}\cdots\theta_m^{2N_m}} \\&\ge \frac{2^{\frac{m-1}{2}}\exp(f(N_1,N_2,\cdots,N_m))}{2^{NH(p_1\cdots p_m)}\theta_1^{N_1}\cdots\theta_m^{N_m}}\\
&\ge \frac{2^{\frac{m-1}{2}}\exp\left(\frac{1}{12N+1}-\frac{1}{24N}+(1/25-1/12)m\right)}{2^{-ND_{\mathrm{KL}}(p||\theta)}} \\
&=2^{\frac{m-1}{2}}2^{ND_{\mathrm{KL}}(p||\theta)}2^{\log{e}(\frac{1}{12N+1}-\frac{1}{24N}+(1/25-1/12)m)}.
\end{align*}
Notice that $\frac{1}{12N+1}-\frac{1}{24N}$ goes to zero when $N\to \infty$, $\frac{m-1}{2}>(1/25-1/12)m$ and $D_{\mathrm{KL}}(P||\theta)\ge 0$. Therefore in this case, 
$$\frac{\binom{N}{N_1,N_2\cdots N_m}\theta_1^{N_1}\cdots\theta_m^{N_m}}{\binom{2N}{2N_1,2N_2\cdots 2N_m}\theta_1^{2N_1}\cdots\theta_m^{2N_m}} \ge 1.$$
When one of $\{N_i\}_{i=1}^N$ equals to $N$, without loss of generality, we assume that $N_1=N$. We have
{\allowdisplaybreaks
$$\frac{\binom{N}{N_1,N_2\cdots N_m}\theta_1^{N_1}\cdots\theta_m^{N_m}}{\binom{2N}{2N_1,2N_2\cdots 2N_m}\theta_1^{2N_1}\cdots\theta_m^{2N_m}}=\frac{1}{\theta_1^{N_1}\cdots\theta_m^{N_m}}>1.$$}%
When there are $k$ numbers out of $N_1, N_2, \cdots, N_m$ that equal to zero, we can simply remove these values and consider the case with alphabet size $m-k$. And this will yield the same result.
\subsection{Joint distributions induce by Laplace and KT probability assignments}\label{sec:joint-dist-lap-kt}
In this section, we re-derive the well-known joint distributions of Laplace and KT probability assignments for completeness. Firstly we consider the Laplace probability assignment. Recalled that we defined 
\begin{equation*}
 q_{\Lap}(X_{j+1} = i| X^j = x^j) = \frac{N_i(x^j) + 1 }{j + m} \quad \text{ for each } i \in [m].
\end{equation*}
Consider a sequence $x^N$ with the number of symbol $i$ given by $N_i$ for each $i\in [m]$. It follows that
\begin{align}
    q_{\Lap}(x^N)&=\prod_{j=0}^{N-1} q_{\Lap}(X_{j+1}=x_{j+1}|X^j=x^{j})\nonumber\\
    &=\prod_{j=0}^{N-1} \frac{N_{x_{j+1}}(x^j) + 1 }{j + m}\nonumber\\
    &\stackrel{(a)}{=}\frac{N_1!N_2!\cdots N_m!}{m(m+1)\cdots(m+N-1)},\label{eqn:laplace-joint}
\end{align}
To see (a), we notice that there are $N_i$'s of symbol $i$ appear in the sequence $x^N$. As a result, terms $1, 2,\ldots N_i$ will appear in the numerator for each $i\in [m]$.
By dividing $N!$ from both numerator and denominator in~\eqref{eqn:laplace-joint}, we finally get 
\[
q_{\Lap}(x^N) \defeq \frac{N_1!N_2!\cdots N_m!}{N!}\cdot \frac{1}{{N+m-1 \choose m-1}}.
\]

Now, we consider KT probability assignment. Recalled that
\begin{equation*}
    q_{\mathrm{KT}}(X_{j+1} = i| X^j = x^j) = \frac{N_i(x^j) + 1/2 }{j + m/2}  \quad \text{ for each } i \in [m].
\end{equation*}
Consider a sequence $x^N$ with the number of symbol $i$ given by $N_i$ for each $i\in [m]$. It follows that
\begin{align}
    q_{\mathrm{KT}}(x^N)&=\prod_{j=0}^{N-1} q_{\mathrm{KT}}(X_{j+1}=x_{j+1}|X^j=x^{j})\nonumber\\
    &=\prod_{j=0}^{N-1} \frac{N_i(x^j) + 1/2 }{j + m/2}\nonumber\\
    &=\prod_{j=0}^{N-1} \frac{2N_i(x^j) + 1 }{2j + m}\nonumber\\
    &\stackrel{(b)}{=}\frac{(2N_1-1)!!(2N_2-1)!!\cdots (2N_n-1)!!}{m(m+2)\cdots (m+2N-2)}\label{eqn:kt-joint}.
\end{align}
To see (b), we notice that there are $N_i$'s of symbol $i$ appear in the sequence $x^N$. As a result, terms $1, 3,\ldots 2N_i-1$ appear in the numerator for each $i\in [m]$.
\bibliographystyle{IEEEtran}
\bibliography{bibliography.bib}

\end{document}

%% file: defns.tex
\usepackage{xspace}
\usepackage{bbm}
\input{mathlig}
\usepackage{mathtools}

\newcommand{\muspace}{\mspace{1mu}}

\DeclareRobustCommand{\scond}{\mathchoice{\muspace\vert\muspace}{\vert}{\vert}{\vert}}
\mathlig{|}{\scond}

\DeclareRobustCommand{\discint}{\mathchoice{\mspace{-1.5mu}:\mspace{-1.5mu}}{\mspace{-1.5mu}:\mspace{-1.5mu}}{:}{:}}
\mathlig{::}{\discint}
\newcommand{\suchthat}{\mathchoice{\colon}{\colon}{:\mspace{1mu}}{:}}

\newcommand{\Ac}{\mathcal{A}}

\newcommand{\Wv}{{\bf W}}

\newcommand{\kt}{{\tilde{k}}}

\def\d{\delta}
\def\e{\epsilon}
\def\eps{\epsilon}

\DeclareMathOperator\E{\mathsf{E}}
\let\P\relax
\DeclareMathOperator\P{\mathsf{P}}

\newcommand{\Bern}{\mathrm{Bern}}

\def\textiid{i.i.d.\@\xspace}
\newcommand\iid{\ifmmode\text{ i.i.d. } \else \textiid \fi}

\newcommand*{\defeq}{\mathrel{\vcenter{\baselineskip0.5ex \lineskiplimit0pt\hbox{\scriptsize.}\hbox{\scriptsize.}}}%
                     =}

\def\mathllap{\mathpalette\mathllapinternal}
\def\mathllapinternal#1#2{%
  \llap{$\mathsurround=0pt#1{#2}$}}

\def\clap#1{\hbox to 0pt{\hss#1\hss}}
\def\mathclap{\mathpalette\mathclapinternal}
\def\mathclapinternal#1#2{%
  \clap{$\mathsurround=0pt#1{#2}$}}

\let\oldstackrel\stackrel
\renewcommand{\stackrel}[2]{\oldstackrel{\mathclap{#1}}{#2}}

\DeclarePairedDelimiterX{\divx}[2]{(}{)}{%
  #1\delimsize\|#2%
}

\def\dtv{\delta}

\renewcommand{\hbar}{h\mathllap{\overline{\vphantom{h}\hphantom{\rule{4.6pt}{0pt}}}\mspace{0.77mu}}}

\catcode`~=11 
\newcommand{\urltilde}{\kern -.06em\lower -.06em\hbox{~}\kern .02em}
\catcode`~=13 

%% file: mathlig.tex
\catcode`@=11
\chardef\mathlig@atcode\count255

\def\actively#1#2{\begingroup\uccode`\~=`#2\relax\uppercase{\endgroup#1~}}
\def\mathlig@gobble{\afterassignment\mathlig@next@cmd\let\mathlig@next= }
\def\mathlig@delim{\mathlig@delim}
\def\mathlig@defcs#1{\expandafter\def\csname#1\endcsname}
\def\mathlig@let@cs#1#2{\expandafter\let\expandafter#1\csname#2\endcsname}
\def\mathlig@appendcs#1#2{\expandafter\edef\csname#1\endcsname{\csname#1\endcsname#2}}

\def\mathlig#1#2{\mathlig@checklig#1\mathlig@end\mathlig@defcs{mathlig@back@#1}{#2}\ignorespaces}

\def\mathlig@checklig#1#2\mathlig@end{%
 \expandafter\ifx\csname mathlig@forw@#1\endcsname\relax
 \expandafter\mathchardef\csname mathlig@back@#1\endcsname=\mathcode`#1%
 \mathcode`#1"8000\actively\def#1{\csname mathlig@look@#1\endcsname}%
 \mathlig@dolig#1\mathlig@delim
\fi
\mathlig@checksuffix#1#2\mathlig@end
}

\def\mathlig@checksuffix#1#2\mathlig@end{%
\ifx\mathlig@delim#2\mathlig@delim\relax\else\mathlig@checksuffix@{#1}#2\mathlig@end\fi
}
\def\mathlig@checksuffix@#1#2#3\mathlig@end{%
\expandafter\ifx\csname mathlig@forw@#1#2\endcsname\relax\mathlig@dosuffix{#1}{#2}\fi
\mathlig@checksuffix{#1#2}#3\mathlig@end
}

\def\mathlig@dosuffix#1#2{%
\mathlig@appendcs{mathlig@toks@#1}{#2}%
\mathlig@dolig{#1}{#2}\mathlig@delim
}

\def\mathlig@dolig#1#2\mathlig@delim{%
 \mathlig@defcs{mathlig@look@#1#2}{%
 \mathlig@let@cs\mathlig@next{mathlig@forw@#1#2}\futurelet\mathlig@next@tok\mathlig@next}%
 \mathlig@defcs{mathlig@forw@#1#2}{%
  \mathlig@let@cs\mathlig@next{mathlig@back@#1#2}%
  \mathlig@let@cs\checker{mathlig@chck@#1#2}%
  \mathlig@let@cs\mathligtoks{mathlig@toks@#1#2}%
  \expandafter\ifx\expandafter\mathlig@delim\mathligtoks\mathlig@delim\relax\else
  \expandafter\checker\mathligtoks\mathlig@delim\fi
  \mathlig@next
 }%
 \mathlig@defcs{mathlig@toks@#1#2}{}%
 \mathlig@defcs{mathlig@chck@#1#2}##1##2\mathlig@delim{%
  \ifx\mathlig@next@tok##1%
   \mathlig@let@cs\mathlig@next@cmd{mathlig@look@#1#2##1}\let\mathlig@next\mathlig@gobble
  \fi 
  \ifx\mathlig@delim##2\mathlig@delim\relax\else
   \csname mathlig@chck@#1#2\endcsname##2\mathlig@delim
  \fi
 }%
 \ifx\mathlig@delim#2\mathlig@delim\else
  \mathlig@defcs{mathlig@back@#1#2}{\csname mathlig@back@#1\endcsname #2}%
 \fi
}%

\catcode`@\mathlig@atcode